\documentclass{LMCS}
\usepackage{gastex}
\usepackage{amsmath,amsfonts,amssymb}
\usepackage{url}
\usepackage{algorithm,algorithmic}
\usepackage{pstricks,color,graphicx}
\usepackage{enumerate,hyperref}
\newcommand{\set}[1]{\{ #1 \}}

\newcommand{\pair}[2]{\langle #1,#2 \rangle}
\newcommand{\triple}[3]{\langle #1,#2,#3 \rangle}

\newcommand{\Nat}{\ensuremath{\mathbb{N}}}

\newcommand{\Rea}{\ensuremath{\mathbb{R}}}

\newcommand{\powerset}[1]{\mathcal{P}({#1})}

\newcommand{\card}[1]{{\rm card}(#1)}
\newcommand{\amap}{f}

\newcommand{\aset}{X}
\newcommand{\asetbis}{Y}
\newcommand{\asetter}{Z}
\newcommand {\length}[1] {\ensuremath{|#1|}}

\newcommand{\avarprop}{p}

\newcommand{\varprop}{\Prop} 
\newcommand{\aformula}{\phi} 
\newcommand{\aformulabis}{\psi} 
\newcommand{\aformulater}{\varphi} 

\renewcommand{\next}{{\sf X}}
\newcommand{\previous}{{\sf X}^{-1}}
\newcommand{\until}{{\sf U}}
\newcommand{\anoperator}{{\sf O}}
\newcommand{\sometimes}{{\sf F}}
\newcommand{\strictsometimes}{{\mathsf F}^{\mathsf +}}

\newcommand{\always}{{\mathsf G}}
\newcommand{\strictalways}{{\mathsf G}^{\mathsf +}}

 \newcommand{\since}{{\sf S}}

\newcommand{\ltl}{{\rm LTL}}

\newcommand{\subf}[1]{sub(#1)} 

\newcommand{\amodel}{\sigma}

\newcommand{\aletter}{a}

\newcommand{\aalphabet}{\Sigma}
\newcommand{\lang}[1]{{\rm L}(#1)}
\newcommand{\alang}{{\rm L}}

\newcommand{\aautomaton}{{\mathcal A}}

\newcommand{\egdef}{\stackrel{\mbox{\begin{tiny}def\end{tiny}}}{=}} 
\newcommand{\equivdef}{\stackrel{\mbox{\begin{tiny}def\end{tiny}}}{\equivaut}} 
\newcommand{\equivaut}{\;\Leftrightarrow\;}

\newcommand {\pspace} {\textsc{pspace}}
\newcommand {\nlogspace} {\textsc{nlogspace}}
\newcommand {\expspace} {\textsc{expspace}}
\newcommand {\twoexpspace} {\textsc{2expspace}}
\newcommand {\np} {\textsc{np}}
\newcommand {\exptime} {\textsc{exptime}}

\newcommand{\mainlogic}{\ltl(\until,\since)}

\newcommand{\Prop}{{\rm PROP}}

\newcommand{\aordinal}{\alpha}
\newcommand{\aordinalbis}{\beta}
\newcommand{\aordinalter}{\gamma}

\newcommand{\arun}{r}
\newcommand{\alocation}{q}
\newcommand{\locations}{Q}

\newcommand{\AL}{\abasis_{lim}}

\newcommand{\ALL}{{\rm all}}
\newcommand{\INFO}{{\rm inf}}

\newcommand{\abasis}{B}

\newcommand{\anelement}{a}

\newcommand{\weight}[1]{{\rm weight}(#1)}

\newcommand{\truncation}[1]{{\rm trunc}_{#1}}

\newcommand{\plusconstant}{2} 

\newcommand{\abstraction}[1]{{\rm abs}(#1)} 

\newcommand{\SAT}{{\rm SAT}}

\newcommand{\defordinal}[1]{{\rm def}_{#1}}

\newcommand{\defstyle}[1]{{\em #1}}

\newif\iflong
\longtrue

\newif\iflpar
\lparfalse

\makeatletter
\let\c@definition\c@theorem
\let\c@lemma\c@theorem
\let\c@corollary\c@theorem
\let\c@proposition\c@theorem
\makeatother

\def\doi{6 (4:9) 2010}
\lmcsheading%
{\doi}
{1--25}
{}
{}
{Feb.~23, 2010}
{Dec.~21, 2010}
{}

\begin{document}

\title[Complexity of LTL over ordinals]{The complexity of 
linear-time temporal logic \\ over the class of ordinals\rsuper*}

\author[S.~Demri]{St\'ephane Demri\rsuper a}	
\address{{\lsuper a}LSV, ENS Cachan, CNRS, INRIA Saclay IdF \\
   \iflong 61, av. Pdt. Wilson, 94235 Cachan Cedex, France}	
\email{demri@lsv.ens-cachan.fr}  
\thanks{{\lsuper a}Partially supported by  project AutoMathA (ESF)}	

\author[A.~Rabinovich]{Alexander Rabinovich\rsuper b}	
\address{{\lsuper b}School of Computer Science \\ 
Tel Aviv University, Ramat Aviv \\
Tel Aviv 69978, Israel}	
\email{rabinoa@post.tau.ac.il}  
\thanks{{\lsuper b}Partially supported by an invited professorship 
from ENS de Cachan.}	



\titlecomment{{\lsuper*}The extended abstract of this paper appeared in~\cite{Demri&Rabinovich07}.}



\begin{abstract}
We consider the temporal logic with since and until modalities.
This temporal  logic  is expressively equivalent over the class
of ordinals to first-order logic by  Kamp's theorem.
We show that it has a \pspace-complete satisfiability problem over the 
class of ordinals. Among the consequences of our proof, we show that 
given the code of some countable ordinal $\aordinal$ and a formula, we can 
decide in \pspace \ whether the formula has a model over $\aordinal$.
In order to show these results, we introduce a  class of simple ordinal
automata, as expressive as B\"uchi 
ordinal automata. The  \pspace \ upper bound 
for the satisfiability problem of the temporal logic is
obtained through a reduction to the nonemptiness problem for the simple ordinal
automata.
\end{abstract}

\maketitle

\section*{Introduction}
\label{section-introduction}

\noindent The main models for time are
$\pair{\Nat}{<}$, the natural
numbers as a model of {\em discrete time} and
the structure $\pair{\Rea}{<}$, the  real line as the model for {\em continuous
time}. These two models are called the {\em canonical models of
time}. A major result concerning linear-time temporal logics is Kamp theorem
\cite{Kamp68,Gabbay&Hodkinson&Reynolds94} which  says that $\mainlogic$, the temporal
logic having ``\emph{Until}'' and ``\emph{Since}'' as only
modalities,
 is expressively
complete for   first-order monadic logic of order
over
the class of
 Dedekind-complete  linear orders.
The canonical models of time
 are indeed Dedekind-complete.
 Another important class of Dedekind-complete orders is the class of  ordinals. 

\noindent In this paper, the satisfiability problem for  the temporal
logic with until and since modalities over the class of ordinals is
investigated. This is the opportunity to generalize what is known about the logic
over $\omega$-sequences. \iflong Our main results are the following.
\begin{enumerate}[(1)]
\item The satisfiability problem for $\mainlogic$ over  the class of
ordinals is \pspace-complete.
%
\item A formula $\aformula$ in $\mainlogic$  has some $\aordinal$-model for some  ordinal
      $\aordinal$
      iff it has an $\aordinalbis$-model for some $\aordinalbis <\omega^{\length{\aformula} + 2}$ where 
$\length{\aformula}$ denotes the size of $\aformula$ for some 
reasonably succinct encoding (see forthcoming 
Corollary~\ref{corollary-small-model}).
\end{enumerate}
\else Our main results are the following: the satisfiability problem
for $\mainlogic$ over  the class of ordinals is \pspace-complete and
a formula $\aformula$ in $\mainlogic$  has some $\aordinal$-model
for some ordinal $\aordinal$
      iff it has an $\aordinalbis$-model for some 
$\aordinalbis <\omega^{\length{\aformula} + 2}$ where $\length{\aformula}$ is the size
of $\aformula$.
\fi

\medskip\noindent In order to prove these results we use an automata-based
approach~\cite{Buchi62,Vardi&Wolper94}. In Section~\ref{section-translation}, we
introduce a new class of ordinal automata  which we call 
\defstyle{simple
ordinal automata}. These automata  are expressive equivalent to
B\"uchi automata over countable ordinals~\cite{Buchi&Siefkes73}. However, the locations and the
transition relations of these automata have additional structures as in~\cite{Rohde97}.
In particular, a location is a subset of a base set $\abasis$.
Herein, we provide a translation from formulae in $\mainlogic$ into
simple ordinal automata that allows to characterize the complexity
of the satisfiability problem for $\mainlogic$.
However, the translation of the formula $\aformula$ into the automaton $\aautomaton_{\aformula}$ provides
an automaton of exponential size in $\length{\aformula}$ but 
the cardinal of the basis of $\aautomaton_{\aformula}$
is linear in $\length{\aformula}$.

\noindent
Section~\ref{section-mainproperty} contains our main technical lemmas. We show
there that every run in a simple ordinal  automaton is equivalent
to a short run.
Consequently, we establish that
a formula $\aformula \in \mainlogic$ has an $\aordinal$-model for some countable ordinal $\aordinal$
iff it has a model of length $\truncation{\length{\aformula}+ \plusconstant}(\aordinal)$ where
$\truncation{\length{\aformula}+ \plusconstant}(\aordinal)$ is a truncated part of $\aordinal$ strictly
less than $\omega^{\length{\aformula}+ \plusconstant} \times 2$ (see the definition of truncation in
Section~\ref{section-mainproperty}).
In Section~\ref{section-nonemptiness} we present two algorithms to solve the
nonemptiness problem for simple ordinal automata.
The first one runs in (simple) exponential time
and does not take advantage of the short run property.
The second algorithm runs in polynomial space  and the
 short run property plays the main role in its
design and  its correctness proof.

In Section~\ref{section-complexity} we investigate several variants
of the satisfiability problem  and show that all of them are
\pspace-complete.
Section~\ref{section-related-work} compares our results with related
works. The  satisfiability problem for $\mainlogic$  over
$\omega$-models is \pspace-complete~\cite{Sistla&Clarke85}. Reynolds
~\cite{Reynolds03,Reynolds??} proved  that the satisfiability
problem for $\mainlogic$  over the reals is \pspace-complete. The
proofs in~\cite{Reynolds03,Reynolds??} are non trivial and difficult
to grasp and it is therefore difficult to compare our proof
technique with those of~\cite{Reynolds03,Reynolds??} even though we
believe cross-fertilization would be fruitful. We provide uniform
proofs and we improve upper bounds for decision problems considered
in \cite{Cachat06,Demri&Nowak07,Rohde97}, see also~\cite{Bezem&Langholm&Walicki07}.
 We  also compare our results and techniques  with Rohde's
thesis~\cite{Rohde97}. Finally we show how our results entail most
of the  results from~\cite{Demri&Nowak07} and we solve some open
problems stated there.

\section{Linear-Time Temporal Logic with Until and Since}
\label{section-definition}

\iflong
\subsection{Basic definitions on ordinals}
Let us start smoothly by recalling 
basic definitions and properties about 
ordinals, see e.g.~\cite{Rosenstein82}
for additional material.
An \defstyle{ordinal} is a totally ordered set which is \defstyle{well ordered},
i.e.\ all its non-empty subsets have a least element. Order-isomorphic
ordinals are considered equal.
They can be more conveniently defined inductively by:
the empty set (written $0$) is an ordinal,
if $\aordinal$ is an ordinal, then 
$\aordinal \cup \{ \aordinal \}$ (written $\aordinal + 1$) is an ordinal
and, 
if $\aset$ is a set of ordinals, then $\bigcup_{\aordinal \in \aset} \aordinal$ is 
an ordinal. 
The ordering is obtained by $\aordinalbis < \aordinal$ iff $\aordinalbis \in \aordinal$. 
An ordinal $\aordinal$ is a \defstyle{successor} ordinal iff there exists an
ordinal $\aordinalbis$ such that $\aordinal = \aordinalbis + 1$.
An ordinal which is not $0$ or a successor ordinal, is a \defstyle{limit}
ordinal.
 The first limit ordinal is written $\omega$.
 Addition, multiplication and exponentiation can be defined on ordinals
 inductively: 
 $\aordinal + 0  =  \aordinal$, $\aordinal + (\aordinalbis + 1)  =  
 (\aordinal + \aordinalbis) + 1$
 and $\aordinal + \aordinalbis  =  \mbox{\it sup} \{  \aordinal + \aordinalter : \aordinalter < \aordinalbis \}$
 where $\aordinalbis$ is a limit ordinal.
 Multiplication and exponentiation are defined similarly. 
Whenever $\aordinal \leq \aordinalbis$, there is a unique ordinal $\aordinalter$
such that $\aordinal + \aordinalter = \aordinalbis$ and we write 
$\aordinalbis - \aordinal$ to denote
$\aordinalter$. 

\subsection{Temporal logic}
\fi 

The formulae of $\mainlogic$ are defined as follows:
$$
\aformula ::= \avarprop \ \mid \ 
              \neg \aformula \ \mid \ 
              \aformula_1 \wedge \aformula_2 \ \mid \
              \aformula_1 \until \aformula_2 \ \mid \
              \aformula_1 \since \aformula_2
$$
where $\avarprop \in \varprop$ for some 
\iflong countably infinite \fi set
$\varprop$ of atomic propositions.
Given a formula $\aformula$ in $\mainlogic$,  
we write $\subf{\aformula}$ to denote
the set of subformulae of $\aformula$ or their negation assuming
that $\neg \neg \aformulabis$ is identified with $\aformulabis$. 
The size of $\aformula$ is defined as the cardinality of
$\subf{\aformula}$ and therefore implicitly we encode formulae as DAGs, which
is exponentially more succinct that the representation by trees.
This feature will be helpful for defining translations that increase only polynomially
the number of subformulae but for which the tree representation might suffer an exponential
blow-up. We use the following standard abbreviations
$\always \aformula = \aformula \wedge \neg (\top \until \neg \aformula)$,
$\strictalways \aformula = \neg (\top \until \neg \aformula)$,
$\sometimes \aformula = \neg \always \neg \aformula$,
$\strictsometimes \aformula = \neg \strictalways \neg \aformula$,
$\next \aformula = \perp \until \aformula$ and $\previous \aformula = \perp \since \aformula$
 that do cause only a polynomial
increase in size.

An \defstyle{$\aordinal$-model} $\amodel$ is a 
function $\amodel: \aordinal \rightarrow \powerset{\varprop}$ for some
ordinal $\aordinal \neq 0$.
The \defstyle{satisfaction relation} 
``$\aformula$ holds in the $\aordinal$-model 
$\amodel$ at position $\aordinalbis$'' ($\aordinalbis < \aordinal$)
is defined as follows:
\begin{enumerate}[$\bullet$]
\item $\amodel, \aordinalbis \models \avarprop$ iff $\avarprop \in \amodel(\aordinalbis)$,

\item $\amodel, \aordinalbis \models \neg \aformula$ iff not $\amodel, \aordinalbis \models \aformula$,
\iflong
\item $\amodel, \aordinalbis \models \aformula_1 \wedge \aformula_2$ iff 
        $\amodel, \aordinalbis \models \aformula_1$ and $\amodel, \aordinalbis \models \aformula_2$,
\else
\ \ $\amodel, \aordinalbis \models \aformula_1 \wedge \aformula_2$ iff 
        $\amodel, \aordinalbis \models \aformula_1$ and $\amodel, \aordinalbis \models \aformula_2$,
\fi 

\item $\amodel, \aordinalbis \models \aformula_1 \until \aformula_2$ iff
      there is $\aordinalter \in (\aordinalbis,\aordinal)$ such that 
       $\amodel, \aordinalter \models \aformula_2$ and
      for every $\aordinalter' \in  (\aordinalbis,\aordinalter)$, we have
      $\amodel, \aordinalter' \models \aformula_1$,

\item $\amodel, \aordinalbis \models \aformula_1 \since \aformula_2$ iff
      there is $\aordinalter  \in [0,\aordinalbis)$ such that 
       $\amodel, \aordinalter \models \aformula_2$ and
      for every $\aordinalter' \in  (\aordinalter,\aordinalbis)$, we have
      $\amodel, \aordinalter' \models \aformula_1$.
\end{enumerate}
Observe that $\since$ and $\until$ are strict ``since'' and ``until'' modalities. 

The (initial) satisfiability problem for $\mainlogic$ consists in determining,
given a formula $\aformula$, whether there is a model $\amodel$ such that
$\amodel, 0 \models \aformula$. 
Note that $\aformula$ is satisfiable in a  model $\amodel$
iff  $\sometimes \aformula$  is
initially satisfiable in $\amodel$. Therefore, there is a polynomial-time
reduction from the satisfiability problem to the initial
satisfiability problem. From now on, we will deal only with the
initial satisfiability problem and for the sake of brevity we will
call it ``satisfiability problem''.

We recall that well orders are particular cases of Dedekind complete linear orders.
Indeed, a chain is Dedekind complete iff every non-empty bounded subset has a
least
upper bound. Kamp's theorem applies herein.

\iflong
\begin{thm} \label{theorem-preliminaries} \cite{Kamp68} 
$\mainlogic$ over the class of ordinals is as expressive as
the first-order logic.
\end{thm}

Moreover, satisfiability for $\mainlogic$ is known to be decidable and as stated below we can restrict
ourselves to countable models. 

\begin{thm} \label{theorem-preliminariesbis} \
\begin{enumerate}[\em(I):] 
\item \cite{Buchi&Siefkes73}
The satisfiability problem for $\mainlogic$ over the class
of countable ordinals is decidable.
\item (see e.g.~\cite[Lemma 6]{Gurevich&Shelah85}) 
A formula in $\mainlogic$ is satisfiable iff it is 
satisfiable in a model of length some
countable ordinal.
\end{enumerate}
\end{thm}
\else
\begin{theorem} (I) \cite{Kamp68} 
$\mainlogic$ over the class of ordinals is as expressive as
the first-order logic over the class of structures $\pair{\aordinal}{<}$ where
$\aordinal$ is an ordinal.
(II) \cite{Buchi&Siefkes73}
The satisfiability problem for $\mainlogic$ over the class
of ordinals is decidable.
\end{theorem}
\fi


Observe that in~\cite{Buchi&Siefkes73} it was proved
that monadic second-order logic over the class of countable
ordinals is decidable and in~\cite{Gurevich&Shelah85} it was shown
that if a formula of the first-order monadic logic is satisfiable in a model
over an ordinal then it is satisfiable in a model over a countable ordinal.
(I) and (II) are immediate consequences of these results and 
the fact that $\mainlogic$ can be easily translated into first-order logic.

Consequently, $\mainlogic$ over the class of ordinals is certainly 
a fundamental logic to be studied. 
We recall below
a central complexity result that we will extend to all ordinals.  

\begin{thm} \cite{Sistla&Clarke85}
Satisfiability for $\mainlogic$, restricted to 
$\omega$-models,  is \pspace- complete.
\end{thm}



\section{Translation from formulae to simple ordinal automata}
\label{section-translation}

\noindent In Section~\ref{section-simple-ordinal-automata}, we
introduce a new class of ordinal automata which we call simple ordinal
automata. These automata are expressive equivalent to B\"uchi automata
over ordinals~\cite{Buchi&Siefkes73}.  However, the locations and the
transition relations of these automata have additional structures.  In
Section~\ref{section-translation-to-automata}, we provide a
translation from $\mainlogic$ into simple ordinal automata which
assigns to every formula in $\mainlogic$ an automaton that recognizes
exactly its models.  We borrow the automata-based approach for
temporal logics from~\cite{Vardi&Wolper94,Kupferman&Vardi&Wolper00}.


\subsection{Simple ordinal automata}
\label{section-simple-ordinal-automata}


\begin{defi} A \defstyle{simple ordinal automaton} $\aautomaton$ is
a structure $\triple{\abasis}{Q}{\delta_{next},\delta_{lim}}$
such that
\begin{enumerate}[$\bullet$]
\iflpar
\item $\abasis$ is a finite set (the \defstyle{basis} of $\aautomaton$),
      $Q \subseteq \powerset{\abasis}$ (the set of locations),
\else
\item $\abasis$ is a finite set (the \defstyle{basis} of $\aautomaton$),
\item $Q \subseteq \powerset{\abasis}$ (the set of \defstyle{locations}),
\fi 
\item $\delta_{next} \subseteq Q \times Q$ is the 
      \defstyle{next-step transition relation},
\item $\delta_{lim} \subseteq \powerset{\abasis} \times Q$ is the 
      \defstyle{limit transition relation}.
\end{enumerate} 
\end{defi}

$\aautomaton$ can be viewed as a finite directed graph whose set of nodes
is structured. 
Limit transitions, whose interpretation is given below,
allow  reaching a node after an infinite amount of steps. 
Given a simple ordinal automaton $\aautomaton$, 
an $\aordinal$-path (or simply a path)
is a map 
$\arun: \aordinal \rightarrow
Q$  for some $\aordinal > 0$ such that
\begin{enumerate}[$\bullet$]
\item for every $\aordinalbis + 1 < \aordinal$,
      $\pair{\arun(\aordinalbis)}{\arun(\aordinalbis+1)} \in \delta_{next}$,
\item for every limit ordinal $\aordinalbis < \aordinal$, 
      $\pair{\AL(\arun,\aordinalbis)}{\arun(\aordinalbis)} \in \delta_{lim}$
      where 
     $$
     \AL(\arun,\aordinalbis) \egdef
     \set{\anelement \in \abasis:
     \exists \ \aordinalter < \aordinalbis \ {\rm such \ that \ for \ every} \
     \aordinalter' \in (\aordinalter,\aordinalbis), \ \anelement \in \arun(\aordinalter')
     }.
     $$
\end{enumerate} 

\medskip\noindent The set $\AL(\arun,\aordinalbis)$ contains exactly the elements of the basis that belong
to every location from some $\aordinalter < \aordinalbis$ until
$\aordinalbis$. We sometimes write $\AL(\arun)$ instead of $\AL(\arun, \aordinal)$ when
$\aordinal$ is a limit ordinal.

Given an $\aordinal$-path $\arun$,  
for $\aordinalbis, \aordinalbis' < \aordinal$ 
we write 
\begin{enumerate}[$\bullet$]
\item $\arun_{\geq \aordinalbis}$ to denote the restriction of $\arun$ 
       to positions
      greater or equal to $\aordinalbis$,
\item $\arun_{\leq \aordinalbis}$ to denote the restriction of
       $\arun$ to positions
      less or equal to $\aordinalbis$,
\item $\arun_{[\aordinalbis,\aordinalbis')}$ to denote the restriction of $\arun$ 
to positions in $[\aordinalbis,\aordinalbis')$ (half-open interval). 
\end{enumerate}
A simple ordinal automaton with \defstyle{acceptance conditions} 
is a structure of the form
$$\triple{\abasis,Q}{I,F,\mathcal{F}}{\delta_{next},\delta_{lim}}$$
where 
\begin{enumerate}[$\bullet$]
\item $I \subseteq Q$ is the set of \defstyle{initial} locations, 
\item $F \subseteq Q$ is the set of \defstyle{final} locations for accepting 
      runs whose length is some successor ordinal,
\item $\mathcal{F} \subseteq \powerset{\abasis}$ encodes the accepting condition for
runs whose length is some limit ordinal.
\end{enumerate}
Given a simple ordinal automaton with acceptance conditions, 
an \defstyle{accepting run}  is a path 
$\arun: \aordinal \rightarrow
Q$  such that 
\begin{enumerate}[$\bullet$]
\item $\arun(0) \in I$, 
\item if $\aordinal$ is a successor ordinal, then $\arun(\aordinal -1) \in F$,
      otherwise $\AL(\arun) \in \mathcal{F}$. 
\end{enumerate}
The \defstyle{nonemptiness problem} for simple ordinal automata consists in checking whether
$\aautomaton$ has an accepting run. 
Our current definition for simple ordinal automata
does not make them language acceptors since they have no alphabet.
It is possible to add in the definition a finite alphabet $\aalphabet$ and to define the next-step
transition relation as a subset of $Q \times \aalphabet \times Q$,
see an example on the right-hand side of Figure~\ref{figure-examples}.
Additionally, the current definition
can be viewed as the case either when the alphabet is a singleton or when
the read letter is encoded in the locations through the dedicated elements of the basis.
This second reading will be in fact used implicitly in the sequel. 
We also write $\aautomaton$ to denote either a simple ordinal automaton
or its extension with acceptance conditions. 

\subsection{Relationships with B\"uchi automata}

Simple ordinal automata with acceptance conditions and alphabet define the same class of
languages as standard ordinal automata in the sense of~\cite{Buchi64,Buchi65}. Main arguments are 
provided below for the sake of completeness. However, we do not need this correspondence in our forthcoming
developments. The main interest for our model of simple ordinal automata rests on the fact that it allows us
to obtain the promised \pspace \ upper bound.  
A \defstyle{standard ordinal automaton} is a structure 
$\aautomaton = \triple{\aalphabet,\locations,I,F,\mathcal{F}}{\delta_{next}}{\delta_{lim}}$  such that
\begin{enumerate}[$\bullet$]
\item $\aalphabet$ is a finite alphabet, 
\item $\locations$ is a finite set of locations,
\item $\delta_{next} \subseteq \locations \times \aalphabet \times \locations$
      and $\delta_{lim} \subseteq \powerset{Q} \times Q$,
\item $I,F \subseteq \locations$ and $\mathcal{F} \subseteq \powerset{\locations}$.
\end{enumerate}
A word $u: \aordinal \rightarrow \aalphabet$ is \defstyle{accepted} by $\aautomaton$
iff there is $\arun: \aordinal \rightarrow \locations$ such that
\begin{enumerate}[$\bullet$]
 \item for every $\aordinalbis+1 < \aordinal$,
       $\pair{\arun(\aordinalbis),u(\aordinalbis)}{\arun(\aordinalbis+1)} \in \delta_{next}$,
 \item for every limit ordinal $\aordinalbis < \aordinal$, 
       $\pair{\INFO(\arun,\aordinalbis)}{\arun(\aordinalbis)} \in \delta_{lim}$
       where 
      $$
      \INFO(\arun,\aordinalbis) \egdef
      \set{\alocation  \in Q:
      {\rm for \ all}  \ \aordinalter < \aordinalbis \ {\rm there \ is} \
       \aordinalter' \in (\aordinal,\aordinalbis) \ {\rm such \ that} \
       \arun(\aordinalter') = \alocation
     }.
     $$
     As usual, $\INFO(\arun,\aordinalbis)$ denotes the set of locations that appear cofinally 
     before
     $\aordinalbis$.
\item $\arun(0) \in I$ and if $\aordinal$ is a successor ordinal, 
      then $\arun(\aordinal -1) \in F$,
      otherwise $\INFO(\arun,\aordinal) \in \mathcal{F}$. 
\end{enumerate}
We write $\lang{\aautomaton}$ to denote the set of words accepted by $\aautomaton$.
Similar definitions can be given for simple ordinal automata with acceptance conditions and alphabet.

 \begin{lem} \label{lemma-correspondence} \
\begin{enumerate}[\em(I):]
\item Given a simple ordinal automaton $\aautomaton$, there is a standard ordinal automaton
$\aautomaton'$ such that $\lang{\aautomaton} = \lang{\aautomaton'}$. 
\item Given a standard ordinal automaton $\aautomaton$, there is a simple ordinal automaton
$\aautomaton'$ such that $\lang{\aautomaton} = \lang{\aautomaton'}$. 
\end{enumerate}
 \end{lem}
 
 \begin{proof}
(I) Let $\aautomaton$ be a simple ordinal automaton 
$\aautomaton = \triple{\aalphabet,\abasis,\locations,I,F,\mathcal{F}}{\delta_{next}}{\delta_{lim}}$.  
We consider the standard ordinal automaton $\aautomaton'$ of the form
$\triple{\aalphabet,\locations,I,F,\mathcal{F}}{\delta_{next}}{\delta_{lim}'}$
such that $\pair{\asetbis}{\alocation} \in \delta_{lim}'$ iff there is a limit transition 
 $\pair{\asetter}{\alocation}
 \in \delta_{lim}$
 satisfying the conditions below.
 \begin{enumerate}[$\bullet$]
 \item for every $\alocation' \in \asetbis$, we have $\asetter \subseteq \alocation'$,
 \item for every element $\anelement \in (\abasis \setminus \asetter)$, there is 
 $\alocation' \in \asetbis$
      such that $\anelement \not \in \alocation'$. 
 \end{enumerate}
 One can easily check that  $\lang{\aautomaton} = \lang{\aautomaton'}$. 
 Observe that  $\aautomaton'$ can be exponentially larger 
 than $\aautomaton$.\\
\noindent
(II)
 Let $\aautomaton =  \triple{\aalphabet,\locations,I,F,\mathcal{F}}{\delta_{next}}{\delta_{lim}}$ 
be a standard ordinal automaton.
We build a simple ordinal automaton 
$\aautomaton' = \triple{\aalphabet,\abasis', 
\locations',I',F',\mathcal{F}'}{\delta_{next}'}{\delta_{lim}'}$
as follows. 
 \begin{enumerate}[$\bullet$]
 \item $\abasis' = \powerset{\locations}$.
\item $\locations' = \set{\aset \in \powerset{\abasis'}: \exists \ \alocation \in \locations, \  
                  \aset = \set{\asetbis \in \powerset{\locations}: \alocation \in \asetbis}}$.
      Below, when $\alocation \in \locations$, by abusing notation, we also write
      $\alocation$ to denote the corresponding location in $\locations'$ equal to 
      $\set{\asetbis \in \powerset{\locations}: \alocation \in \asetbis}$. 
\item $I' = I$, $F' = F$ and $\mathcal{F}' = \mathcal{F}$. 
 \item For $\aletter \in \aalphabet$ and $\alocation, \alocation' \in \locations$, 
      $\pair{\alocation,\aletter}{\alocation'} \in \delta_{next}'$ only if
       $\pair{\alocation,\aletter}{\alocation'}
        \in \delta_{next}$.
 \item For $\asetbis' \subseteq \abasis'$ and $\alocation \in \locations$, 
      $\pair{\asetbis'}{\alocation} \in \delta_{lim}'$ only if  there is a limit transition
        $\pair{\asetbis}{\alocation} \in \delta_{lim}$ such that
        $\asetbis' = \set{\anelement \in \abasis': \asetbis \subseteq \anelement}$.
 \end{enumerate}
Again, one can easily check that  $\lang{\aautomaton} = \lang{\aautomaton'}$.
\end{proof} 

Let $\alang_0$ be the set of words $u: \aordinal \rightarrow \set{0,1}$
such that for $\aordinalbis < \aordinal$, $\aordinalbis = \omega^2 
\aordinalter$ for some ordinal $\aordinalter$ iff 
$u(\aordinalbis) = 1$. 
The left-hand side of Figure~\ref{figure-examples} presents
a standard ordinal automaton (with three locations) accepting $\alang_0$.
Next-step transitions are represented by plain arrows whereas 
limit transitions are represented by dashed arrows. Moreover,
$F = 
\set{\alocation_1, \alocation_{\omega}, \alocation_{\geq \omega^2}}$
and $\mathcal{F} = \powerset{\set{\alocation_1, \alocation_{\omega}, \alocation_{\geq \omega^2}}}$. 
The right-hand side of Figure~\ref{figure-examples} presents
a corresponding simple ordinal automaton along the lines of the
proof of Lemma~\ref{lemma-correspondence}. Its 
basis $\abasis$ is equal
to $\powerset{\set{\alocation_1, \alocation_{\omega}, 
\alocation_{\geq \omega^2}}}$ and we write 
$\mathbf{\alocation_1}$ to denote
$\set{\set{\alocation_1}, \set{\alocation_1,\alocation_{\omega}}, \set{\alocation_1,\alocation_{\geq \omega^2}}, 
\set{\alocation_1,\alocation_{\omega},\alocation_{\geq \omega^2}}}$.
$\mathbf{\alocation_{\omega}}$ and 
$\mathbf{\alocation_{\geq \omega^2}}$ 
are defined 
similarly.

\begin{figure}
\includegraphics[scale=0.5,angle=-90]{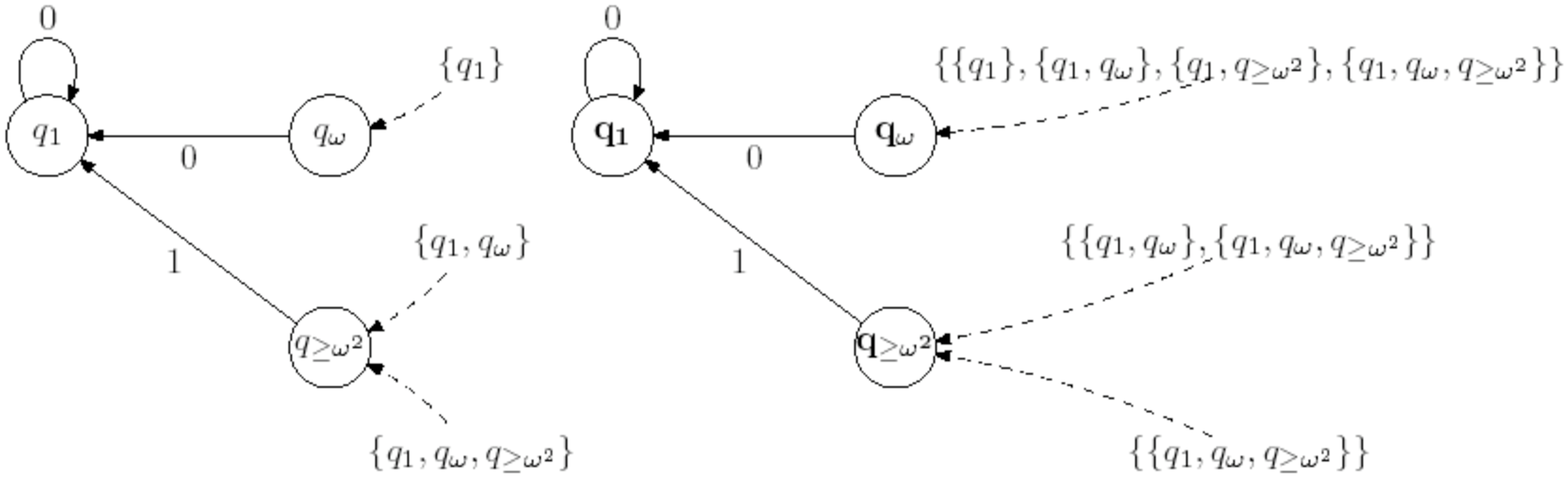}
\vspace{-0.5in}
\caption{Ordinal automata accepting $\alang_0$}
\label{figure-examples}
\end{figure}

\subsection{Translation from \texorpdfstring{$\mainlogic$}{LTL(U,S)} formulae to simple ordinal automata}
\label{section-translation-to-automata} 

As usual, a set $\asetbis$ 
is a \defstyle{maximally Boolean consistent} subset of 
$\subf{\aformula}$ when the following
conditions are satisfied:
\iflong
\begin{enumerate}[$\bullet$]
\item for every $\aformulabis \in \subf{\aformula}$,
      $\neg \aformulabis \in \asetbis$ iff $\aformulabis \not \in \asetbis$,
\item for every $\aformulabis_1 \wedge \aformulabis_2 \in \subf{\aformula}$,
      $\aformulabis_1 \wedge \aformulabis_2 \in \asetbis$ 
      iff $\aformulabis_1, \aformulabis_2 \in \asetbis$. 
\end{enumerate}
\else
for every $\aformulabis \in \subf{\aformula}$,
      $\neg \aformulabis \in \asetbis$ iff $\aformulabis \not \in \asetbis$ and
for every $\aformulabis_1 \wedge \aformulabis_2 \in \subf{\aformula}$,
      $\aformulabis_1 \wedge \aformulabis_2 \in \asetbis$ iff $\aformulabis_1, 
\aformulabis_2 \in \asetbis$. 
\fi
Given a formula $\aformula$, the simple ordinal automaton 
$\aautomaton_{\aformula} = 
\triple{\abasis,Q}{I,F,\mathcal{F}}{\delta_{next},\delta_{lim}}$ is defined as follows:
\begin{enumerate}[$\bullet$]
\item $\abasis = \subf{\aformula}$.
\item $Q$ is the set of maximally Boolean consistent subsets of
      $\subf{\aformula}$.
\item $I$ is the set of locations that contain $\aformula$ 
     and no  elements of the form  $\aformulabis_1 \since \aformulabis_2$.
\item $F$ is the set of locations with no elements of the form $\aformulabis_1 
\until \aformulabis_2$.
\item 
A subset $\asetbis$ of $\abasis$ is in $\mathcal{F}$ if there are no $\aformulabis_1$ and $\aformulabis_2$
such that  $\set{\aformulabis_1, \neg \aformulabis_2, 
     \aformulabis_1 \until \aformulabis_2} \subseteq \asetbis$.
\item For all $\alocation, \alocation' \in Q$,
      $\pair{\alocation}{\alocation'} \in \delta_{next}$ $\equivdef$ the conditions
     below are satisfied:
     \begin{enumerate}[\hbox to8 pt{\hfill}]
     \item\noindent{\hskip-12 pt\bf (next$_{\until}$):}\ 
          for \iflong every \fi $\aformulabis_1 \until \aformulabis_2 \in 
           \subf{\aformula}$,
           $\aformulabis_1 \until \aformulabis_2 \in \alocation$ iff
           either $\aformulabis_2 \in \alocation'$ or 
           $\aformulabis_1, 
            \aformulabis_1 \until \aformulabis_2 \in \alocation'$,
     \item\noindent{\hskip-12 pt\bf (next$_{\since}$):}\ 
          for \iflong every \fi $\aformulabis_1 \since \aformulabis_2 \in 
          \subf{\aformula}$,
           $\aformulabis_1 \since \aformulabis_2 \in \alocation'$ iff
           either $\aformulabis_2 \in \alocation$ or 
           $\aformulabis_1, 
            \aformulabis_1 \since \aformulabis_2 \in \alocation$.
     \end{enumerate}


\item For all $\asetbis \subseteq \abasis$ and $\alocation \in Q$,
      $\pair{\asetbis}{\alocation} \in \delta_{lim}$ $\equivdef$ the conditions
      below are satisfied:
\begin{enumerate}[\hbox to8 pt{\hfill}]
     \item\noindent{\hskip-12 pt\bf(lim$_{\until}$1):}\ 
     if $\aformulabis_1, \neg \aformulabis_2, 
     \aformulabis_1 \until \aformulabis_2 \in \asetbis$, then
     either $\aformulabis_2 \in \alocation$ or 
            $\aformulabis_1, \aformulabis_1 \until \aformulabis_2
            \in \alocation$,
     \item\noindent{\hskip-12 pt\bf(lim$_{\until}$2):}\ 
     if $\aformulabis_1,
      \aformulabis_1 \until \aformulabis_2 
     \in \alocation$ and $\aformulabis_1 \in \asetbis$, 
     then $\aformulabis_1 \until \aformulabis_2 
     \in \asetbis$,
     \item\noindent{\hskip-12 pt\bf(lim$_{\until}$3):}\ 
     if $\aformulabis_1 \in \asetbis$,
                                 $\aformulabis_2 \in \alocation$ and
                         $\aformulabis_1 \until \aformulabis_2$ is in the basis
                         $\abasis$, 
                         then
                              $\aformulabis_1 \until \aformulabis_2 \in 
                               \asetbis$,
     \item\noindent{\hskip-12 pt\bf(lim$_{\since}$):}\ for every
       $\aformulabis_1 \since \aformulabis_2 \in \subf{\aformula}$,
       $\aformulabis_1 \since \aformulabis_2 \in \alocation$ iff
       ($\aformulabis_1 \in \asetbis$ and $\aformulabis_1 \since
       \aformulabis_2 \in \asetbis$).
     \end{enumerate}

%
\end{enumerate}

Even though the conditions above are compatible with the intuition that
a location contains the formulae that are promised to be satisfied, at the current stage it might sound
mysterious how the conditions have been made up (mainly for the conditions related to 
limit transitions). For some of them, their justification comes with
the proof of Lemma~\ref{lemma-correctness}.

Let $\amodel$ be an  $\aordinal$-model   and 
$\aformula$ be a formula in $\mainlogic$.
The \defstyle{Hintikka sequence} for  $\amodel$ and $\aformula$ is an  
$\aordinal$-sequence $H^{\amodel,\aformula}$ defined as follows:
for every $\aordinalbis < \aordinal$,
$$ H^{\amodel, \aformula} (\aordinalbis) \egdef
\{\aformulabis \in \subf{\aformula} ~:~
 \amodel ,\aordinalbis \models \aformulabis\}.
$$
Given a run $\arun: \aordinal \rightarrow Q$, we write $mod(\arun): \aordinal \rightarrow
\powerset{\varprop}$ to denote the $\aordinal$-model $\amodel$ defined as follows:
$\amodel(\aordinalbis) \egdef \set{\avarprop \in \varprop: \avarprop \in \arun(\aordinalbis)}$.
It is clear that if $\arun$ is an Hintikka sequence for $\amodel$ and $\aformula$, then
$mod(\arun) = \amodel$. 

Now we can state the correctness lemma.

 \begin{lem} \label{lemma-correctness} \
\begin{enumerate}[\em(I):]
\item If $\amodel, 0 \models \aformula$,
then 
\iflpar $H^{\amodel,\aformula}$
\else the Hintikka sequence for $\amodel$ and $\aformula$
\fi 
 is an accepting run of $\aautomaton_{\aformula}$.
\item If $\arun$ is an accepting run of $\aautomaton_{\aformula}$,
 then  $mod(\arun), 0 \models \aformula$
and $\arun$ is the Hintikka sequence for $mod(\arun)$ and $\aformula$.
\item $\aformula$ is 
satisfiable  iff $\aautomaton_{\aformula}$ has an
accepting run.
\end{enumerate}
\end{lem}

\iflong
\proof
First, (III) is an immediate consequence of (I) and (II).

\begin{enumerate}[(I):]

\item Suppose that there is a model $\amodel: \aordinal \rightarrow 
\powerset{\varprop}$ (with $\aordinal > 0$) such that $\amodel, 0 \models \aformula$. 
By using  $\mainlogic$ semantics, it is straightforward to check that 
$H^{\amodel, \aformula}$ is accepted by $\aautomaton_{\aformula}$.

\item Let $\arun: \aordinal \rightarrow Q$ be an accepting run of $\aautomaton_{\aformula}$.
Let us show by structural induction that for all $\aformulabis
\in \subf{\aformula}$ and $\aordinalbis < \aordinal$, we have $mod(\arun), \aordinalbis 
\models \aformulabis$ iff $\aformulabis \in \arun(\aordinalbis)$. The base case and the cases with
Boolean operators in the induction step are by an easy verification.
 The only interesting cases in the induction step are related to the temporal operators
$\until$ and $\since$.
 Below, let $\amodel$ be
$mod(\arun)$.

\medskip\noindent
{\em Case $\until$}: $\aformulabis = \aformulabis_1 \until \aformulabis_2$. \\
Let us reason \emph{ad absurdum}. 
Suppose that $\set{\aordinalbis: \aformulabis \in \arun(\aordinalbis)}
\neq \set{\aordinalbis: \amodel, \aordinalbis \models \aformulabis}$. 
Let $\aordinalbis$ be the smallest ordinal which belongs to only one of these sets.
We consider two cases: 
($\amodel, \aordinalbis \models \aformulabis$ and $\aformulabis \not \in \arun(\aordinalbis)$)
 -- Case I below --
or 
 ($\aformulabis  \in \arun(\aordinalbis)$ and $\amodel, 
\aordinalbis \not \models \aformulabis$)
-- Case II below.
\begin{enumerate}[\hbox to8 pt{\hfill}]
\item\noindent{\hskip-12 pt\bf Case I:}\ 
Let $\aordinalter$ be the smallest ordinal verifying $\aordinalbis < \aordinalter <
\aordinal$, $\amodel, \aordinalter \models \aformulabis_2$ and
for every $\aordinalter' \in (\aordinalbis,\aordinalter)$,
we have $\amodel, \aordinalter' \models \aformulabis_1 \wedge \neg \aformulabis_2$. 
By induction hypothesis, $\aformulabis_2 \in \arun(\aordinalter)$ and
for every $\aordinalter' \in  (\aordinalbis,\aordinalter)$,
$\set{\aformulabis_1,\neg \aformulabis_2} \subseteq \arun(\aordinalter')$. 

First, we are going to show that $\neg \aformulabis\in  \arun(
\aordinalter')$ for every   $ \aordinalter'\in [\aordinalbis ,\aordinalter)$.
This is true for $\aordinalbis$. Assume that this is true for $\aordinalbis'$
then it is true for $\aordinalbis'+1$ by condition (next$_{\until}$).  
Assume that
$\aordinalter'$  is a limit ordinal and $\aformulabis \not \in
\arun(\aordinalbis')$ for every 
$\aordinalbis' \in [\aordinalbis ,\aordinalter')$.
Then, by  condition (lim$_{\until}$2) we obtain that 
$\neg \aformulabis \in  \arun(\aordinalter')$.
Next, consider two cases:

\noindent {\em Case a):} $\aordinalter$ is a successor, say
$\aordinalter=\aordinalbis'+1$. We have $ \neg \aformulabis
\in \arun(\aordinalbis')$ and $\aformulabis_2
\in  \arun(\aordinalter)$. 
This contradicts condition (next$_{\until}$).

\noindent {\em Case b):} $\aordinalter$ is a limit ordinal. In
this case $ \{\neg\aformulabis,\aformulabis_1\}\subseteq  
\AL(\arun,\aordinalter)$ and $\aformulabis_2\in \arun(\aordinalter)$. 
This contradicts condition (lim$_{\until}$3).\\
%
%

\item\noindent{\hskip-12 pt\bf Case II:}\ 
  Now suppose that $\aformulabis_1 \until \aformulabis_2 \in \arun(\aordinalbis)$ and
$\amodel, \aordinalbis \not \models \aformulabis_1 \until \aformulabis_2$.\medskip
 
\noindent
{\em Case a\/}): For every $\aordinalter$ such that $\aordinalter \in (\aordinalbis,\aordinal)$,
we have $\amodel, \aordinalter \not \models \aformulabis_2$ ($\aformulabis_2$ does not hold
on $\amodel$ strictly after $\aordinalbis$). 

By induction hypothesis, for 
every $\aordinalter \in (\aordinalbis,\aordinal)$,
$\neg \aformulabis_2 \in \arun(\aordinalter)$. Let us show that for
every $\aordinalter \in (\aordinalbis,\aordinal)$,
$\set{\aformulabis_1, \aformulabis_1 \until \aformulabis_2}
\subseteq \arun(\aordinalter)$.

\medskip{\em Base case}: $\aordinalter = \aordinalbis + 1$. \\
By condition (next$_{\until}$), 
$\aformulabis_1 \until \aformulabis_2 \in \arun(\aordinalbis)$
and $\neg \aformulabis_2 \in \arun(\aordinalbis+1)$ imply
$\set{\aformulabis_1, \aformulabis_1 \until \aformulabis_2}
\subseteq \arun(\aordinalbis+1)$.

\medskip{\em Induction step}: 
\begin{enumerate}[$\bullet$]
\itemsep 0 cm
\item if $\aordinalter = \aordinalter' + 1$, then by
      condition (next$_{\until}$),
     $\aformulabis_1 \until \aformulabis_2 \in \arun(\aordinalter')$
and $\neg \aformulabis_2 \in \arun(\aordinalter'+1)$ imply
$\set{\aformulabis_1, \aformulabis_1 \until \aformulabis_2}
\subseteq \arun(\aordinalter'+1)$.
\item if $\aordinalter$ is a limit ordinal,
     by induction hypothesis, $\set{\aformulabis_1, \neg \aformulabis_2, 
     \aformulabis_1 \until \aformulabis_2} \subseteq \AL(\arun,\aordinalter)$.
     By condition (lim$_{\until}$1), $\set{\aformulabis_1, \aformulabis_1 \until \aformulabis_2}
\subseteq \arun(\aordinalter)$ since $\aformulabis_2 \not \in \arun(\aordinalter)$. 
\end{enumerate}

Consequently, if $\aordinal$ is a limit ordinal, then 
$\set{\aformulabis_1, \neg \aformulabis_2, 
     \aformulabis_1 \until \aformulabis_2} \subseteq \AL(\arun,\aordinal)$
which is in contradiction with the definition of $\mathcal{F}$ in $\aautomaton_{\aformula}$.
Similarly, if $\aordinal = \aordinal'+1$, then
$\aformulabis_1 \until \aformulabis_2 \in \arun(\aordinal')$ which 
is in contradiction with the definition of $F$.

\medskip\noindent
{\em Case b\/}): There is a minimal ordinal $\aordinalter \in (\aordinalbis,\aordinal)$ such that 
$\amodel, \aordinalter \models \neg \aformulabis_1 \wedge \neg \aformulabis_2$
and for every $\aordinalter' \in (\aordinalbis,\aordinalter)$,
we have $\amodel, \aordinalter' \models \aformulabis_1 \wedge \neg \aformulabis_2$.
By induction hypothesis, $\set{\neg \aformulabis_1, \neg \aformulabis_2} 
\subseteq \arun(\aordinalter)$ and 
 for every $\aordinalter' \in (\aordinalbis,\aordinalter)$,
$\set{\aformulabis_1, \neg \aformulabis_2} \subseteq \arun(\aordinalter')$.
Let us show that for
every  $\aordinalter' \in (\aordinalbis,\aordinalter)$,
$\set{\aformulabis_1, \aformulabis_1 \until \aformulabis_2}
\subseteq \arun(\aordinalter')$.

\medskip\noindent
{\em Base case}: $\aordinalter' = \aordinalbis + 1$. \\
By condition (next$_{\until}$), 
$\aformulabis_1 \until \aformulabis_2 \in \arun(\aordinalbis)$
and $\neg \aformulabis_2 \in \arun(\aordinalbis+1)$ imply
$\set{\aformulabis_1, \aformulabis_1 \until \aformulabis_2}
\subseteq \arun(\aordinalter')$.

\medskip{\em Induction step}: 
\begin{enumerate}[$\bullet$]
\item If $\aordinalter' = \aordinalter'' + 1$, then
     by condition (next$_{\until}$),
     $\aformulabis_1 \until \aformulabis_2 \in \arun(\aordinalter'')$
and $\neg \aformulabis_2 \in \arun(\aordinalter''+1)$ imply
$\set{\aformulabis_1, \aformulabis_1 \until \aformulabis_2}
\subseteq \arun(\aordinalter''+1)$.
\item If $\aordinalter'$ is a limit ordinal, then
     by induction hypothesis, $\set{\aformulabis_1, \neg \aformulabis_2, 
     \aformulabis_1 \until \aformulabis_2} \subseteq \AL(\arun,\aordinalter')$.
     By condition (lim$_{\until}$1), $\set{\aformulabis_1, \aformulabis_1 \until \aformulabis_2}
\subseteq \arun(\aordinalter')$ since $\aformulabis_2 \not \in \arun(\aordinalter')$. 
\end{enumerate}

Consequently, if $\aordinalter$ is a limit ordinal, then 
$\set{\aformulabis_1, \neg \aformulabis_2, 
     \aformulabis_1 \until \aformulabis_2} \subseteq \AL(\arun,\aordinalter)$
which leads to a contradiction by condition (lim$_{\until}$1).
Indeed, by induction hypothesis,  $\set{\neg \aformulabis_1, \neg \aformulabis_2} 
\subseteq \arun(\aordinalter)$. 
Similarly, if $\aordinalter = \aordinalter'+1$, then
$\aformulabis_1 \until \aformulabis_2 \not \in \arun(\aordinalter')$ which leads
to a contradiction by condition (next$_{\until}$).\\
\end{enumerate}

\noindent 
{\em Case $\since$}: $\aformulabis = \aformulabis_1 \since \aformulabis_2$. \\
Let us reason \emph{ad absurdum}. 
Suppose that $\set{\aordinalbis: \aformulabis \in \arun(\aordinalbis)}
\neq \set{\aordinalbis: \amodel, \aordinalbis \models \aformulabis}$. 
Let $\aordinalbis$ be the smallest ordinal that belongs to only one of these sets.
Again, we distinguish two cases, namely either 
($\amodel, \aordinalbis \models \aformulabis$ and $\aformulabis \not \in \arun(\aordinalbis)$)
-- Case I below --
or 
 ($\aformulabis  \in \arun(\aordinalbis)$ and $\amodel, \aordinalbis \not \models \aformulabis$)
-- Case II below.

\begin{enumerate}[\hbox to8 pt{\hfill}]
\item\noindent{\hskip-12 pt\bf Case I:}\ 
So $\aordinalbis > 0$ and
there is $\aordinalter \in [0, \aordinalbis)$ such that $\amodel, \aordinalter
\models \aformulabis_2$ and for every $\aordinalter' \in (\aordinalter,\aordinalbis)$,
we have $\amodel, \aordinalter' \models \aformulabis_1$.
By induction hypothesis, $\aformulabis_2 \in \arun(\aordinalter)$ and 
for every $\aordinalter' \in (\aordinalter,\aordinalbis)$,
$\aformulabis_1 \in \arun(\aordinalter')$.
Observe that for every $\aordinalter' \in (\aordinalter,\aordinalbis)$,
we have $\amodel, \aordinalter' \models \aformulabis$ and $\aformulabis \in \arun(\aordinalter')$
($\aordinalbis$ is minimal). 
\begin{enumerate}[$\bullet$]
\item If $\aordinalbis = \aordinalbis' + 1$ then by condition (next$_{\since}$)
       $\aformulabis_2 \not \in \arun(\aordinalbis')$  and
      $\set{\aformulabis_1, \aformulabis_1 \since \aformulabis_2}
      \not \subseteq \arun(\aordinalbis')$. 
      If $\aordinalter = \aordinalbis'$, then this leads to a contradiction since 
      $\aformulabis_2 \in \arun(\aordinalter)$. 
      Similarly, if $\aordinalter < \aordinalbis'$, then 
      $\aformulabis \not \in \arun(\aordinalbis')$ since $\aformulabis_1 \in 
       \arun(\aordinalbis')$.
      Since $\amodel, \aordinalbis' \models \aformulabis_1 \since \aformulabis_2$, this leads
     to a contradiction by the minimality of $\aordinalbis$. 
\item If $\aordinalbis$ is a limit ordinal, then by condition (lim$_{\since}$)
      either $\aformulabis_1 \not \in \AL(\arun,\aordinalbis)$ or
      $\aformulabis_1 \since \aformulabis_2 \not \in \AL(\arun,\aordinalbis)$.  
      By induction hypothesis, $\aformulabis_1 \in \AL(\arun,\aordinalbis)$.
      Hence, there is $\aordinalbis' \in (\aordinalter,\aordinalbis)$ such that
     $\aformulabis_1 \since \aformulabis_2 \not \in \arun(\aordinalbis')$, which is in contradiction with the
     minimality of $\aordinalbis$. 
\end{enumerate}

\item\noindent{\hskip-12 pt\bf Case II:}\ 

\noindent
{\em Case a\/}): For every $\aordinalter \in [0,\aordinalbis)$, 
               $\amodel, \aordinalbis \not \models \aformulabis_2$. \\
By induction hypothesis,  
for every $\aordinalter \in [0,\aordinalbis)$, $\aformulabis_2 \not \in 
\arun(\aordinalter)$. Moreover, 
for every $\aordinalter \in [0,\aordinalbis)$, we have
 $\amodel, \aordinalbis \not \models \aformulabis_1 \since \aformulabis_2$. 
\begin{enumerate}[$\bullet$]
\item If $\aordinalbis = \aordinalbis' + 1$ then by condition (next$_{\since}$), 
$\set{\aformulabis_1,\aformulabis_1 \since \aformulabis_2} \subseteq \arun(\aordinalbis')$
which leads to a contradiction by  minimality of $\aordinalbis$. 
\item If $\aordinalbis$ is a limit ordinal, then $\set{\aformulabis_1, 
\aformulabis_1 \since \aformulabis_2} \subseteq \AL(\arun,\aordinalbis)$ by
condition (lim$_{\since}$). Hence, for some $\aordinalbis' < \aordinalbis$,
$\aformulabis_1 \since \aformulabis_2 \in \arun(\aordinalbis')$, 
which leads again to a contradiction by the minimality of $\aordinalbis$. 
\item If $\aordinalbis = 0$, then we also have a contradiction since
$\arun(0)$ does not contain any since formulae. Observe that in the previous
case analyses with ordinals, the case ``0'' has been irrelevant. 
\end{enumerate}

\medskip{\em Case b\/}): $\amodel, \aordinalbis \not \models \aformulabis$ and not a). \\
Remember that $\aformulabis \in \arun(\aordinalbis)$. 
There is $\aordinalbis' < \aordinalbis$ such that $\amodel, \aordinalbis' 
\not \models \aformulabis_1$. Otherwise, by induction hypothesis and by not a),
we have $\amodel, \aordinalbis \models \aformulabis$, a contradiction.\\
{\em Case b.1:} There is a maximal position $\aordinalter < \aordinalbis$ such that
$\amodel, \aordinalter \not \models \aformulabis_1$. \\
For every $\aordinalter' \in (\aordinalter,\aordinalbis)$,
we have $\amodel, \aordinalter \not \models \aformulabis_2$, otherwise 
$\amodel, \aordinalbis \models \aformulabis$ which would lead to a contradiction.
Let us show by transfinite induction that for every 
$\aordinalter' \in (\aordinalter,\aordinalbis]$,
$\aformulabis \not \in \arun(\aordinalter')$.

\medskip{\em Base case:} $\aordinalter' = \aordinalter + 1$. \\
$\neg \aformulabis_1, \neg \aformulabis_2 \in \arun(\aordinalter)$ imply
by condition (next$_{\since}$) that $\aformulabis \not \in \arun(\aordinalter')$. 

\medskip{\em Induction step:}
\begin{enumerate}[$\bullet$]
\item If $\aordinalter' = \aordinalter'' + 1$, then 
$\neg \aformulabis_2, \neg \aformulabis \in \arun(\aordinalter'')$ by induction hypothesis.
By condition (next$_{\since}$) $\aformulabis \not \in \arun(\aordinalter')$.
\item If $\aordinalter'$ is a limit ordinal, then $\neg \aformulabis \in
\AL(\arun,\aordinalter')$ and by condition (lim$_{\since}$),
 $\aformulabis \not \in \arun(\aordinalter')$.
\end{enumerate}
Hence, $\aformulabis \not \in \arun(\aordinalbis)$, which leads to a contradiction.

\medskip\noindent 
{\em Case b.2} There is no maximal position $\aordinalter < \aordinalbis$ such that
$\amodel, \aordinalter \not \models \aformulabis_1$ (the most delicate case). \\
Consequently, there is a unique position  $\aordinalter \leq \aordinalbis$
such that for every $\aordinalter' < \aordinalter$, there is 
$\aordinalter' < \aordinalter'' < \aordinalter$ verifying $\neg \aformulabis_1
\in \arun(\aordinalter'')$. 
This means that 
\begin{enumerate}[$\bullet$]
\item for every $\aordinalter' \in [\aordinalter,\aordinalbis]$,
      $\aformulabis_1 \in \arun(\aordinalter')$,
\item $\aformulabis_1 \not \in \AL(\arun,\aordinalter)$ and,
\item  by condition (lim$_{\since}$) $\aformulabis \not \in \arun(\aordinalter)$. 
\end{enumerate}
Moreover, for every $\aordinalter' \in (\aordinalter,\aordinalbis)$,
$\neg \aformulabis_2 \in \arun(\aordinalter')$ otherwise by
induction hypothesis, $\amodel, \aordinalbis \models \aformulabis$, which would lead
to a contradiction.
Let us show by transfinite induction that for every 
$\aordinalter' \in (\aordinalter,\aordinalbis]$,
$\aformulabis \not \in \arun(\aordinalter')$.

\medskip{\em Base case:} $\aordinalter' = \aordinalter + 1$. \\
$\neg \aformulabis, \neg \aformulabis_2 \in \arun(\aordinalter)$ imply
by condition (next$_{\since}$)
$\aformulabis \not \in \arun(\aordinalter')$. 

\medskip{\em Induction step:}
\begin{enumerate}[$\bullet$]
\item If $\aordinalter' = \aordinalter'' + 1$, then 
$\neg \aformulabis_2, \neg \aformulabis \in \arun(\aordinalter'')$ by induction hypothesis.
By condition (next$_{\since}$) $\aformulabis \not \in \arun(\aordinalter')$.
\item If $\aordinalter'$ is a limit ordinal, then $\neg \aformulabis \in
\AL(\arun,\aordinalter')$ and by condition (lim$_{\since}$),
 $\aformulabis \not \in \arun(\aordinalter')$.
\end{enumerate} 
Hence, $\aformulabis \not \in \arun(\aordinalbis)$, which leads to a contradiction.\qed
\end{enumerate}
\end{enumerate}

\else
\iflpar
\else
The proof can be found in Appendix~\ref{section-proof-lemma-correctness}.
\fi 
\fi

\section{Short Run Properties}
\label{section-mainproperty}

\noindent In this section, we establish pumping arguments that are
useful to show that
\begin{enumerate}[$\bullet$]
\item in order to check the satisfiability status of the formula $\aformula$,
      there is no need to consider models of length greater than $\omega^{\length{\aformula} + 2}$,
\item simple ordinal automata cannot distinguish ordinals with identical tails (defined
below precisely with the notion of \defstyle{truncation}). 
\end{enumerate}

Let $\aautomaton$ be a simple ordinal automaton and $\asetbis$ be a subset
of its basis. 
$\asetbis$ is said to be \defstyle{present} in $\aautomaton$ iff either 
 there is a limit transition of the form $\pair{\asetbis}{\alocation}$ in $\aautomaton$
or $\asetbis \in \mathcal{F}$.
Given a set $\asetbis$ present in $\aautomaton$, its \defstyle{weight}, 
noted $\weight{\asetbis}$, 
is the maximal $l$ such that $\asetbis_1 \subset \asetbis_2 \subset \cdots \subset
\asetbis_l$ is a sequence of present subsets in $\aautomaton$ and $\asetbis_1 = \asetbis$
($\subset$ denotes proper subset inclusion).
Obviously, $\weight{\asetbis} \leq \card{\abasis} + 1$.

Given a path $\arun: \aordinal  \rightarrow Q$ in $\aautomaton$ with 
a limit ordinal $\aordinal \geq \omega$, 
its \defstyle{weight}, noted $\weight{\arun}$, is the maximal value in the set
\iflong
$$
\set{\weight{\AL(\arun,\aordinalbis}): \aordinalbis < \aordinal, \ 
     \aordinalbis \ {\rm is \ a \ limit \ ordinal}
} \cup \set{\weight{\AL(\arun)}}
$$
When $\aordinal$ is a successor ordinal, the maximal value is computed only 
from the first set of the above union.
\else
$
\set{\weight{\AL(\arun,\aordinalbis}): \aordinalbis < \aordinal, \ 
     \aordinalbis \ {\rm is \ a \ limit \ ordinal}
}.
$
\fi
By convention, if a path has no limit transition, then its weight is zero
(equivalently, its length is strictly less than $\omega + 1$). 
Furthermore, we write $\ALL(\arun)$ to denote the set 
$$\ALL(\arun) \egdef \bigcap_{\aordinalbis < \aordinal}
\arun(\aordinalbis)$$
that corresponds to the set of elements from the basis that are present
in all locations of the run $\arun$. 
Let $\arun$, $\arun'$ be two paths
of respective length $\aordinal$ and $\aordinal'$.
We say that  $\arun$ and $\arun'$ are \defstyle{congruent}, 
written $\arun \sim \arun'$,  iff the conditions below are meet:
\begin{enumerate}[(1)]
\itemsep 0 cm
\item $\arun(0) = \arun'(0)$.
\item Either both $\aordinal$ and $\aordinal'$ are successor ordinals and 
      $\arun(\aordinal-1) = \arun'(\aordinal'-1)$ or 
      both $\aordinal$ and $\aordinal'$ are limit ordinals and $\AL(\arun) = \AL(\arun')$.
\item $\ALL(\arun) = \ALL(\arun')$. 
\end{enumerate}

Let $\arun_1$ be a path of length $\aordinal$ and
$\arun_2$ be a path of length $\aordinalbis$ such that
if $\aordinal$ is a limit ordinal then
$\pair{\AL(\arun_1)}{\arun_2(0)} \in \delta_{lim}$ otherwise
$\arun_1(\aordinal -1) = \arun_2(0)$. 
The concatenation $\arun_1 \cdot \arun_2$ is the path $\arun$ of length
$\aordinal + \aordinalbis$ such that for $\aordinalter \in [0,\aordinal)$,
$\arun(\aordinalter) = \arun_1(\aordinalter)$ and for $\aordinalter \in [0,\aordinalbis)$,
$\arun(\aordinal + \aordinalter) = \arun_2(\aordinalter)$. 
For every ordinal $\aordinal$, the concatenation of $\aordinal$-sequences
of paths is defined similarly.
The relation $\sim$ is a congruence for the concatenation 
operation
on 
\iflpar
paths. 
\else
paths as stated in details below.\fi

\begin{lem} \label{lemma-congruence}\hfill
\begin{enumerate}[\em(I):]
\item Let $\arun \cdot \arun_0 \cdot \arun'$ and $\arun_1$ be
two paths such that $\arun_0 \sim \arun_1$. Then, 
$\arun \cdot \arun_1 \cdot \arun'$ is a path that is congruent to
$\arun \cdot \arun_0 \cdot \arun'$. 
\item Let $\arun_0^0, \arun_0^1, \arun_0^2, \ldots$ and
$\arun_1^0, \arun_1^1, \arun_1^2, \ldots$  be two $\omega$-sequences
of pairwise consecutive paths 
such that for $i \geq 0$, $\arun_0^i \sim \arun_1^i$ and their length is a successor ordinal. If 
$\arun \cdot (\arun_0^0 \cdot \arun_0^1 \cdot \arun_0^2 \cdot  \ldots) \cdot \arun'$
is a path, then it is congruent to 
$\arun \cdot (\arun_1^0 \cdot \arun_1^1 \cdot \arun_1^2 \cdot  \ldots) \cdot \arun'$.
\end{enumerate} 
\end{lem}
\noindent The proof of the above lemma  is by an easy 
\iflong verification but observe that for the proof of (II)
the third set of equalities from the definition of the 
congruence $\sim$ ensures that $\arun \cdot (\arun_0^0 \cdot \arun_0^1 \cdot \arun_0^2 \cdot  
\ldots) \cdot \arun'$ is a path.
\else verification.
\fi

\begin{lem} \label{lemma-mainproperty}
Let $\arun: \aordinal \rightarrow Q$ be  a
path in $\aautomaton$ for some countable ordinal $\aordinal$ such that if 
$\aordinal$ is a limit ordinal, then $\AL(\arun)$ is present in $\aautomaton$.
Then, there is  a path $\arun': \aordinal'  \rightarrow Q$ for
 $\aordinal' < \omega^{max(1,\weight{\arun}) +1}$ such that $\arun \sim \arun'$ and
$\weight{\arun'} \leq \weight{\arun}$.
\end{lem}

\iflong
\proof
The proof is by induction on the weight of the paths.
When the weight of the path is zero, $\arun' = \arun$ 
already satisfies the condition $\arun \sim \arun'$.  
We only treat below the cases with paths of length some limit ordinals.
The case with paths of length some successor ordinals is similar. 
All the runs $\arun'$ built below satisfy that $\weight{\arun'} \leq \weight{\arun}$
for the following reasons. Indeed, no additional
limit transitions are applied when building $\arun'$ from $\arun$ and when $\arun$ is of length
some limit ordinal, $\AL(\arun) = \AL(\arun')$.
Hence, below we shall not further emphasize 
$\weight{\arun'} \leq \weight{\arun}$. 
 
\medskip\noindent
{\em Base case 1}:  $\weight{\arun} = 1$ and $\aordinal = \omega^2$. \\
There is $n \geq 0$ such that 
\begin{enumerate}[(1)]
\item for every $\anelement \in 
                 \abasis \setminus \ALL(\arun)$, there is $\aordinalter
\leq \omega \cdot n$ such that $\anelement \not \in \arun(\aordinalter)$,
\item $\ALL(\arun_{\geq \omega \cdot n}) = \AL(\arun)$.
\end{enumerate}
The first condition states that if $\anelement$ does not belong to $\ALL(\arun)$, then this is already
witnessed by $\arun_{\leq \omega \cdot n}$. 
Furthermore, in general 
 $\ALL(\arun_{\geq \omega \cdot k}) \subseteq \AL(\arun)$ but the second condition above 
states that for $n$ large enough,
we can obtain $\AL(\arun) \subseteq \ALL(\arun_{\geq \omega \cdot n})$.  

Hence,  $\ALL(\arun_{\leq \omega \cdot n})  = \ALL(\arun)$ 
and  $\AL(\arun_{\geq \omega \cdot (n + 1)}) = \AL(\arun)$.  Besides, 
$\AL(\arun)$ is present in $\aautomaton$. 
Let
$\asetbis_i = \AL(\arun_{\leq \omega \cdot i})$ for $i \geq n+1$. 
By construction of $n$, for all $i \geq n+1$, 
$ \AL(\arun) \subseteq \asetbis_i$. Moreover, $\weight{\AL(\arun)} = 1$.
Hence, for all $i \geq n+1$, 
$\asetbis_i = \AL(\arun)$.
Consequently,  $\arun' : \omega \cdot (n+1)  \rightarrow
Q$ with 
 $\arun'(\aordinalbis) = \arun(\aordinalbis)$
for $\aordinalbis < \omega \cdot (n+1)$ verifies $\arun \sim \arun'$. 
In order to show that $\ALL(\arun) = \ALL(\arun')$ it is sufficient
to observe that $\ALL(\arun) \subseteq \ALL(\arun')$ since $\arun'$ contains less locations
than $\arun$ 
and $\ALL(\arun'_{\leq \omega \cdot n})  = \ALL(\arun)$.

\medskip\noindent
{\em Base case 2}: $\weight{\arun} = 1$ and $\aordinal = \omega^2 \times \aordinalbis$. \\
The proof is by transfinite induction. 
The base case with $\aordinalbis = 1$ is actually the above base case 1. 
Now suppose that 
$\aordinal = \omega^2 \times (\aordinalbis + 1)$.
By induction hypothesis and by the base case 1, there are paths
$\arun': \aordinalter \rightarrow Q$ and $\arun'':\aordinalter'  \rightarrow Q$ such that 
$\arun' \sim \arun_{< \omega^2 \times \aordinalbis}$,
$\arun'' \sim \arun_{\geq \omega^2 \times \aordinalbis}$ and 
 $\aordinalter  + \aordinalter' < \omega^2$.
Consequently, the concatenation of $\arun'$ and $\arun''$ provides a path satisfying
the adequate conditions.

Now suppose that $\aordinal = \omega^2 \times \aordinalbis$ where $\aordinalbis$ is a limit
ordinal. 
Since $\aordinal$ is countable, 
there is an increasing 
sequence $(\aordinalbis_i)_{i \in \Nat}$ of ordinals
strictly smaller than $\aordinalbis$ such that $\aordinalbis_0 = 0$ and
$\aordinalbis = lim \ \aordinalbis_i$ (see e.g.~\cite[Theorem 3.36]{Rosenstein82}). 
%
%
Observe that for every $i$, $\aordinalbis_{i+1} - \aordinalbis_{i} < \aordinalbis$.
Hence, for every $i$, by induction hypothesis, there is a path 
$\arun'_i: \aordinalter_i \rightarrow Q$ such that
$\arun'_i \sim \arun_{[\omega^2 \times \aordinalbis_i, 
       \omega^2 \times \aordinalbis_{i+1})}$ and
$\aordinalter_i < \omega^2$.
Consequently, $\arun_0' \cdot \arun_1' \cdot \arun_2' \cdots$ is
a path of length at most $\omega^2$ congruent to $\arun$ by Lemma~\ref{lemma-congruence}
(the length may be exactly $\omega^2$).
By using again arguments from the base case 1, we obtain a path that satisfies the adequate
conditions.

\medskip\noindent
{\em Base case 3}: $\weight{\arun} = 1$ and $\aordinal = \omega^2 \times \aordinalbis + 
\omega \times n$ ($n \in \Nat$). \\
The existence of a path satisfying the adequate conditions is an immediate consequence
of the base case 2.

\medskip\noindent
{\em Induction case}. 

\noindent
{\em Case 1\/}:  $\aordinal = \omega^{\weight{\arun} + 1}$. \\
There is $n \geq 0$ such that 
\begin{enumerate}[(1)]
\item for every $\anelement \in 
          \abasis \setminus \ALL(\arun)$, there is $\aordinalter
\leq \omega^{\weight{\arun}} \cdot n$ such that $\anelement \not \in \arun(\aordinalter)$,
\item $\ALL(\arun_{\geq \omega^{\weight{\arun}} \cdot n}) = \AL(\arun)$.
\end{enumerate}
Hence,  $\ALL(\arun_{\leq \omega^{\weight{\arun}} \cdot n}) 
           = \ALL(\arun)$ 
and  $\AL(\arun_{\geq \omega^{\weight{\arun}} \cdot (n + 1)}) = \AL(\arun)$.  
Besides, $\AL(\arun)$ is present in $\aautomaton$ and $\weight{\AL(\arun)} \leq \weight{\arun}$. 
If there is a limit ordinal $\aordinalbis \in [\omega^{\weight{\arun}} \cdot n, 
\aordinal)$ such that
 $\AL(\arun_{< \aordinalbis}) =  \AL(\arun)$, then  
 $\arun' : \aordinalbis \rightarrow
 Q$ with 
  $\arun'(\aordinalter) = \arun(\aordinalter)$
 for $\aordinalter < \aordinalbis$ verifies the required conditions.
Otherwise, 
 suppose that for every limit ordinal $\aordinalbis \in [\omega^{\weight{\arun}} \cdot n ,\aordinal)$, 
 we have 
$\AL(\arun_{< \aordinalbis}) \neq  \AL(\arun)$.
By construction of $n$, 
 for every limit ordinal $\aordinalbis$ in $[\omega^{\weight{\arun}} \cdot n, \aordinal)$, 
 $\AL(\arun) \subset \AL(\arun_{< \aordinalbis})$.
By induction hypothesis, for every $i > n$, 
there is a path 
$\arun'_i: \aordinalter_i  \rightarrow Q$ such that
$\arun_{[\omega^{\weight{\arun}} \times i, 
        \omega^{\weight{\arun}} \times (i+1))} 
\sim \arun'_i$ and $\aordinalter_i  < \omega^{\weight{\arun}}$.
Consequently, $\arun'_0 \cdot \arun'_1 \cdot \arun'_2 \cdots$ is a path of length
less than $\omega^{\weight{\arun}}$ that is congruent to $\arun$ by 
Lemma~\ref{lemma-congruence}.

\medskip\noindent
{\em Case 2\/}: $\aordinal = \omega^{\weight{\arun}+1} \times \aordinalbis$. \\
The proof is by transfinite induction as in the base case 2. Indeed, suppose that 
$\aordinal = \omega^{\weight{\arun}+1} \times (\aordinalbis + 1)$.
There are paths
$\arun': \aordinalter  \rightarrow Q$ and $\arun'':\aordinalter'  \rightarrow Q$ such that 
$\arun' \sim \arun_{< \omega^{\weight{\arun}+1} \times \aordinalbis}$,
$\arun'' \sim \arun_{\geq \omega^{\weight{\arun}+1} \times \aordinalbis}$ and 
 $\aordinalter  + \aordinalter'  < \omega^{\weight{\arun}+1}$.
Consequently, the concatenation of $\arun'$ and $\arun''$ provides a path satisfying
the adequate conditions.

Now suppose that $\aordinal = \omega^{\weight{\arun}+1} 
\times \aordinalbis$ where $\aordinalbis$ is a limit
ordinal. Hence, there is an increasing 
sequence $(\aordinalbis_i)_{i \in \Nat}$ of ordinals
strictly smaller than $\aordinalbis$ such that $\aordinalbis_0 = 0$ and
$\aordinalbis = lim \ \aordinalbis_i$ (see e.g.~\cite[Theorem 3.36]{Rosenstein82}). 
Observe that for every $i$, $\aordinalbis_{i+1} - \aordinalbis_{i} < \aordinalbis$.
Hence, for every $i$, by induction hypothesis, there is a path 
$\arun'_i: \aordinalter_i \rightarrow Q$ such that
$\arun'_i \sim \arun_{[\omega^{\weight{\arun}+1} \times \aordinalbis_i, 
       \omega^{\weight{\arun}+1} \times \aordinalbis_{i+1})}$ and
$\aordinalter_i +1 < \omega^{\weight{\arun}+1}$.
Consequently, $\arun_0' \cdot \arun_1' \cdot \arun_2' \cdots$ is
a path of length less than $\omega^{\weight{\arun}+1}$ 
congruent to $\arun$ by Lemma~\ref{lemma-congruence} (the lenght may be equal to $\omega^{\weight{\arun}+1}$).
By using the case 1 in the induction step, we can get a path that satisfies the adequate
conditions.

\medskip\noindent
{\em Case 3\/}: $\aordinal = \omega^{\weight{\arun} + 1} \times \aordinalbis + 
\omega^{\weight{\arun}} \times n_{\weight{\arun}} +  \cdots + 
\omega^{1} \times n_{1}$ with $n_{\weight{\arun}}$, \ldots, $n_1 \in \Nat$. \\
The existence of a path satisfying the required conditions is an
immediate consequence of the case 2.\qed
 
\else
\iflpar
\else The proof can be found in Appendix~\ref{section-proof-lemma-mainproperty}. 
\fi 
\fi
Lemma~\ref{lemma-mainproperty} below states  
a crucial property for most of  complexity results established in the sequel.
Indeed, for usual ordinal automata, it is not possible to get this polynomial
bound as an exponent of $\omega$ for the length of the short paths. 
Actually, the exponent is linear in the cardinal of its basis and can be
logarithmic in the number of locations for large automata. 
By combination of Lemma~\ref{lemma-correctness} and Lemma~\ref{lemma-mainproperty}, we obtain
the following interesting result.

\begin{cor} \label{corollary-small-model}
If $\aformula$ is satisfiable, then $\aformula$ has an  $\aordinal$-model 
with  $\aordinal < \omega^{\length{\aformula} + \plusconstant}$. 
\end{cor}


For $n \in \Nat$, let $\truncation{n}$
be the function that assigns to every ordinal $\aordinal > 0$ an ordinal in $(0,\omega^n 2)$ as follows.  
$\aordinal$ can be written in the form $\aordinal = \omega^n \aordinalter + \aordinalbis$
with $\aordinalbis \in [0,\omega^n)$. Then $\truncation{n}(\aordinal) = 
\omega^n \times min(\aordinalter,1) + \aordinalbis$.
For instance $\truncation{2}(\omega^3) = \omega^2$, 
             $\truncation{2}(\omega^2 + \omega) = \omega^2 + \omega$ 
             and $\truncation{2}(\omega^2 \times 2) = \omega^2$.
The ordinals $\aordinal$, $\aordinalbis$ are \defstyle{$n$-equivalent},
written $\aordinal \approx_n \aordinalbis$, $\equivdef$ $\truncation{n}(\aordinal) =
\truncation{n}(\aordinalbis)$. \\


\begin{lem} \label{lemma-pumping}
Let $\aautomaton$ be a simple ordinal automaton.
\begin{enumerate}[\em(I):]
\item If  $\arun$ is a path
of length $\omega^{\weight{\arun} + 1} \times \aordinal$
for some countable ordinal $\aordinal > 0$, then there is
a path $\arun'$ of length $\omega^{\weight{\arun} + 1}$ such that
$\arun \sim \arun'$ and $\weight{\arun'} \leq \weight{\arun}$.
\item  If a path $\arun$ has length $\omega^{\weight{\arun} + 1}$ and $\weight{\arun} \geq 1$, 
 then for every ordinal $\aordinal > 0$, there is a path $\arun'$
 of length $\omega^{\weight{\arun} + 1} \times \aordinal$ such that
 $\arun \sim \arun'$ and $\weight{\arun'} \leq \weight{\arun}$.
\item If $\arun$ is a path of length some countable ordinal $\aordinal$ and
 $\aordinalbis 
             \approx_{\card{\abasis} + 2} \aordinal$, then
             there is a path $\arun'$ of length $\aordinalbis$ such that
             $\arun \sim \arun'$.  
\end{enumerate}
\end{lem}

Only in (I), the ordinal $\aordinal$ is supposed to be countable. 
\iflong
\proof
(III) is a direct consequence of (I) and (II). 
Indeed, suppose $\aordinal = \omega^{\card{\abasis} + 2} \aordinalter_0 + \aordinalter_1$ and
$\aordinalbis = \omega^{\card{\abasis} + 2} \aordinalter_0' + \aordinalter_1'$
with $ \aordinalter_1 = \aordinalter'_1 \in [0, \omega^{\card{\abasis} + 2})$, and $\aordinalter_0 \geq 1$ iff 
 $\aordinalter_0' \geq 1$. If $\aordinalter_0 = \aordinalter_0' = 0$, 
then $\aordinal = \aordinalbis$ and we are done.
Otherwise ($\weight{\arun} \geq 1$), let $K > 0$ such that $K + \weight{\arun} = \card{\abasis} + 2$. 
Since $\weight{\arun} \leq \card{\abasis} + 1$ such a value $K$ exists and therefore (I) can be applied. 
There is a run $\arun'$ such that $\arun' \sim \arun_{\leq \omega^{\card{\abasis} + 2} \aordinalter_0}$ 
and $\arun'$ is of length 
$\omega^{ \weight{\arun} + 1}$ by (I). 
If $\weight{\arun'} \neq \weight{\arun}$, then we apply again (I) on $\arun'$ in order to obtain
a run $\arun''$ such that $\arun'' \sim \arun'$,  $\arun''$ is of length 
$\omega^{ \weight{\arun'} + 1}$. If again $\weight{\arun''} \neq \weight{\arun'}$, we cannot repeat
this process more than $\card{\abasis} +1$ times. Eventually, we obtain a run
$\arun_0$ such that $\arun_0 \sim \arun_{\leq \omega^{\card{\abasis} + 2} \aordinalter_0}$ 
and $\arun_0$ is of length 
$\omega^{ \weight{\arun_0} + 1}$. 
By (II), there is a run $\arun_1$ such that $\arun_1 \sim \arun_0$
and $\arun_1$ is of length $\omega^{\card{\abasis} + 2} \aordinalter_0'$ by (II). 
Consequently, $\arun_1 \cdot \arun_{\geq \omega^{\card{\abasis} + 2} \aordinalter_0} 
\sim \arun$ and  $\arun_1 \cdot \arun_{\geq \omega^{\card{\abasis} + 2} \aordinalter_0} $ is of length
$\aordinalbis$. 


\begin{enumerate}[(I):]
\item The proof is by transfinite induction
on $\aordinal$. 
Again, all the runs $\arun'$ built below satisfy that $\weight{\arun'} \leq \weight{\arun}$
for the following reasons. Indeed, no additional
limit transitions are applied when building $\arun'$ from $\arun$ and when $\arun$ is of length
some limit ordinal, $\AL(\arun) = \AL(\arun')$.
Hence, below we shall not further emphasize  
$\weight{\arun'} \leq \weight{\arun}$. 
We behave similarly for the proof of (II).

\noindent
Observe that the run $\arun$ cannot be of length $\omega$. In the sequel,
we assume that $\weight{\arun} \geq 1$. 
The  base case
with $\aordinal = 1$ is immediate. Suppose that the induction assertion holds true for
$\aordinal$ and let us show that it holds true for $\aordinal + 1$.  
By Lemma~\ref{lemma-mainproperty}, there is a run $\arun'$ of length
strictly less than $\omega^{\weight{\arun}+1}$ 
such that $\arun' \sim \arun_{< \omega^{\weight{\arun}+1} 
\times \aordinal}$. Hence $\arun' \cdot 
\arun_{\geq \omega^{\weight{\arun}+1} \times \aordinal}
\sim \arun$ and its length is exactly $\omega^{\weight{\arun}+1}$. 
Now suppose that $\aordinal$
is a limit ordinal and for every smaller ordinal, the property holds true.
Let $\arun$ be a run of length $\omega^{\weight{\arun}+1} \times \aordinal$. 
There exists an increasing sequence $(\aordinal_i)_{i \in \Nat}$ with
$\aordinal_0 = 0$ and $\aordinal = lim \ \aordinal_i$ (see e.g.~\cite[Theorem 3.36]{Rosenstein82}).
For $i \geq 0$, let $\aordinal'_i$ be $\omega^{\weight{\arun}+1} \aordinal_i + \omega^{\weight{\arun}}$.
Observe that $\aordinal_i' - \omega^{\weight{\arun}+1} \aordinal_i = \omega^{\weight{\arun}}$ and
$\omega^{\weight{\arun}+1} \aordinal_i < \aordinal'_i < \omega^{\weight{\arun}+1} \aordinal_{i+1}$. 
For $i \geq 0$, let $\aordinalbis_i$ be 
$\omega^{\weight{\arun}+1} \times \aordinal_i$.
For every $i \geq 0$, let $\arun_i$ be
the path $\arun_{[\aordinal_i', \aordinalbis_{i+1})}$. 
By Lemma~\ref{lemma-mainproperty}, for every $j \geq 0$, there is
a path $\arun'_{j}$ congruent to $\arun_{j}$ of length
strictly less than $\omega^{\weight{\arun}+1}$ and $\weight{\arun'_j} \leq
\weight{\arun_j}$. 
Let $\arun'$ be the run
$\arun_{[\aordinalbis_0, \aordinal_0')} \arun'_0 \arun_{[\aordinalbis_1, \aordinal_1')} \arun'_1 
 \arun_{[\aordinalbis_2, \aordinal_2)} \arun'_2 \ldots$.
The path $\arun'$ is  exactly of length 
$\omega^{\weight{\arun}+1}$ and it is congruent to $\arun$.

The proof is by double induction on the weight and on
$\aordinal$.

\medskip\noindent{\em Base case:} $\weight{\arun} = 1$. \\
Let $\amap: [\Nat]^2 \rightarrow Q \times \powerset{Q} \times Q$ be the  
function whose domain is made of  unordered pairs $\set{i,j}$ of natural numbers (say, $i < j$) 
such that 
$$
\amap(\set{i,j}) \egdef
\triple{\arun(\omega \times i)
}{
\AL(\arun_{[\omega \times i, \omega \times j)})
}{\arun(\omega \times j)
}
$$
By Ramsey's Theorem  (see e.g. \cite{Ramsey30,Rosenstein82}), 
there is an infinite set $\asetbis \subseteq \Nat$ such that 
$\amap$ restricted to $[\asetbis]^2$ is constant.
Hence, there is a value $\triple{\alocation^{\star}}{A}{\alocation^{\star}}$
and an infinite sequence $0 \leq i_0 < i_1 < i_2 < \cdots$ such that for every $k \geq 0$,
we have $\amap(\set{i_k,i_{k+1}}) = \triple{\alocation^{\star}}{A}{\alocation^{\star}}$. 
Observe that $A = \AL(\arun)$ and for every $k$, 
we also have $\AL(\arun) \subseteq \AL(\arun_{[\omega \times i_k, \omega \times i_{k+1})})$.
Since $\weight{\arun} = 1$, we get that 
$\AL(\arun) = \AL(\arun_{[\omega \times i_k, \omega \times i_{k+1})})$.

Let us come back to the proof by induction. 
The base case with $\aordinal = 1$ is immediate. 
Suppose that the property holds true for
$\aordinal$ and let us show that it holds true for $\aordinal + 1$.  
By induction hypothesis, there is a path $\arun'$ congruent to $\arun$ of length 
$\omega^2 \times \aordinal$. 
Since $\arun_{\geq \omega \times i_0}$ is also a path of
length $\omega^2$, $\arun' \cdot  \arun_{\geq \omega \times i_0}$
is a path ($\AL(\arun_{< \omega \times i_0}) = A$), it 
 is congruent to $\arun$ and its length is precisely $(\omega^2 \times \aordinal) + \omega^2$. 

Now suppose that $\aordinal$
is a limit ordinal and for every smaller ordinal, the property holds true.
There exists a strictly increasing sequence $(\aordinal_i)_{i \in \Nat}$ with
$\aordinal_0 = 0$ and $\aordinal = lim \ \aordinal_i$.
By the induction hypothesis
there is a run $\arun_j$ of length  $\omega^{2} \times (\aordinal_{j+1} - \aordinal_{j})$ 
congruent to
$\arun_{\geq \omega \times i_j}$ ($\arun_{\geq \omega \times i_j}$ is also of length $\omega^2$).
Then, $\arun_0 \cdot \arun_1 \cdot \arun_2 \cdots$ is congruent to $\arun$
and it is of length $\omega^2 \times \aordinal$. Observe that
$\AL(\arun_0 \cdot \arun_1 \cdot \arun_2 \cdots )$ is precisely $A$ that is equal
to $\AL(\arun)$, as stated above.

\medskip\noindent
{\em Induction step:}: $\weight{\arun} > 1$ and the property holds for all the paths of
weight strictly less than $\weight{\arun}$. \\
The base case with $\aordinal = 1$ is immediate.

\begin{enumerate}[$\bullet$]
\item Suppose that the property holds true for
$\aordinal$ and let us show that it holds true for $\aordinal + 1$. 
As in the base case, we define a coloring function $\amap$ such that we color
the interval with endpoints at positions of the form
$\omega^{\weight{r}} \times n$. 
Similarly to the base case, there is a triple 
$\triple{\alocation^{\star}}{A}{\alocation^{\star}}$ 
and a 
sequence $0 \leq i_0 < i_1 < i_2 < \cdots$ such that for 
every $k \geq 0$,
$\amap(\set{i_k,i_{k+1}}) = \triple{\alocation^{\star}}{A}{\alocation^{\star}}$.
If there is $\aordinalbis < \omega^{\weight{\arun} + 1}$ such that
$\AL(\arun_{< \aordinalbis}) = \AL(\arun)$ then by induction
hypothesis, there is $\arun' \sim \arun$ such that
$\arun'$ is of length $\omega^{\weight{\arun} + 1} \times \aordinal$ and 
$\AL(\arun') = \AL(\arun_{< \aordinalbis})$. 
Hence $\arun' \cdot \arun_{\geq \aordinalbis}$ is a path, 
$\arun' \cdot \arun_{\geq \aordinalbis} \sim \arun$ and its length
is $\omega^{\weight{\arun} + 1} \times (\aordinal + 1)$.
If there is no such an ordinal $\aordinalbis$, for every limit ordinal
$\aordinalbis \in [\omega^{\weight{\arun}} \times (i_1-1),\omega^{\weight{\arun}} \times i_1)$,
$\AL(\arun) \subset \AL(\arun_{\aordinalbis})$
 since $A \subset \AL(\arun_{< \aordinalbis})$. 
Hence $W = \weight{\arun_{[\omega^{\weight{\arun}} \times (i_1-1), 
\omega^{\weight{\arun}} \times i_1)}} < \weight{\arun}$.
By the induction hypothesis there is a run $\arun' \sim  
\arun_{[\omega^{\weight{\arun}} \times (i_1-1), 
\omega^{\weight{\arun}} \times i_1)}$ 
of length $\omega^{W+1} \times (\omega^{(\weight{\arun} - W)} \times \aordinal)$,
that is of length $\omega^{\weight{\arun} + 1} \times \aordinal$ by associativity of multiplication.
Hence $\arun_{< \weight{\arun} \times (i_1-1)} 
\cdot \arun' \cdot \arun_{\geq \omega^{\weight{\arun}} \times i_1} \sim \arun$ and it is of length 
$\omega^{\weight{\arun} + 1} \times (\aordinal+1)$.
\item Now suppose that $\aordinal$
is a limit ordinal and for every smaller ordinal, the property holds true.
There exists a strictly increasing sequence $(\aordinal_i)_{i \in \Nat}$ with
$\aordinal_0 = 0$ and $\aordinal = lim \ \aordinal_i$. 
As above,  a triple of the form  $\triple{\alocation^{\star}}{A}{\alocation^{\star}}$
and an $\omega$-sequence $i_0 < i_1 < i_2 < \ldots$ can be defined. 
Observe that for every $k \geq 1$, for every limit ordinal
$\aordinalbis \in [\omega^{\weight{\arun}} \times (i_k-1),\omega^{\weight{\arun}} \times i_k)$, 
$\AL(\arun) \subseteq \AL(\arun, \aordinalbis)$
since $A = \AL(\arun)$ and $A \subseteq \AL(\arun, \aordinalbis)$.
Hence the weight of  $\arun_{[\omega^{\weight{\arun}} \times (i_k-1), 
\omega^{\weight{\arun}} \times i_k)}$, noted $W_k$, is less or equal to $\weight{\arun}$. 
By induction hypothesis, for every $k \geq 1$, there is a path
$\arun_k \sim \arun_{[\omega^{\weight{\arun}} \times (i_k-1),  
\omega^{\weight{\arun}} \times i_k)}$ of length 
 $\omega^{W_k+1} \times (\omega^{(\weight{\arun} - W_k)} \times (\aordinal_{k+1} - \aordinal_{k})$,
that is of length $\omega^{\weight{\arun} + 1} (\aordinal_{k+1} - \aordinal_{k})$.
Hence  $\arun' = \arun_1 \arun_2 \arun_3 \ldots$ is path, it is congruent to $\arun$ and of length
 $\omega^{\weight{\arun} + 1} \times \aordinal$.  It is worth observing that $\AL(\arun') = A$.\qed
\end{enumerate}
\end{enumerate}

\else
\iflpar
\else The proof can be found in Appendix~\ref{section-proof-lemma-pumping}.
\fi
\fi
Because of the translation from formulae to automata, we can also
establish a pumping lemma at the level of formulae.


\begin{lem} \label{lemma-pumping-formulae} \ 
\begin{enumerate}[\em(I):]
\item Let $\aautomaton$ be a simple ordinal automaton with acceptance conditions
           and $\aordinal$, $\aordinalbis$ be countable ordinals 
           such that $\aordinal \approx_{\card{\abasis}+2}
           \aordinalbis$. 
      Then, $\aautomaton$ has an accepting run of length $\aordinal$ iff
      $\aautomaton$ has an accepting run of length $\aordinalbis$. 
\item Let $\aformula$ be a formula in $\mainlogic$
and $\aordinal$, $\aordinalbis$ be countable ordinals such that $\aordinal \approx_{\length{\aformula} + 2}
           \aordinalbis$.
Then $\aformula$ has an $\aordinal$-model  iff 
$\aformula$ has a $\aordinalbis$-model.
\end{enumerate}
\end{lem}

\proof\hfill
\begin{enumerate}[(I):]
\item Direct consequence of Lemma~\ref{lemma-mainproperty} and 
Lemma~\ref{lemma-pumping} since accepting runs can be viewed as paths.

\item By Lemma~\ref{lemma-correctness},  
$\aformula$ has an $\aordinal$-model  iff 
$\aautomaton_{\aformula}$ has an accepting run $\arun$
of length $\aordinal$. Since the cardinal of the basis of $\aautomaton_{\aformula}$ is precisely $\length{\aformula}$,
by (I)  we get
that $\aautomaton_{\aformula}$ has an accepting run $\arun$
of length $\aordinal$ iff
$\aautomaton_{\aformula}$ has an accepting run $\arun$
of length $\aordinalbis$. Equivalently, 
$\aformula$ has a $\aordinalbis$-model. 
\qed
\end{enumerate}

\section{Checking nonemptiness of simple ordinal automata}
\label{section-nonemptiness}

\noindent In this section, we provide algorithms to check whether a
simple ordinal automaton admits accepting runs.  The first one runs in
exponential time.  Our optimal algorithm runs in polynomial space in
the size of the basis (see Section~\ref{section-polynomial-space}).

\iflpar
\else \subsection{An exponential-time algorithm for checking nonemptiness}
\fi 

Let $\aautomaton$ be a simple ordinal automaton 
$\triple{\abasis}{Q,I,F,\mathcal{F}}{\delta_{next},\delta_{lim}}$. 
We provide below an algorithm to check given $q,q' \in Q$ and $n \in \Nat$
whether there is path $\arun: \aordinal + 1 \rightarrow Q$ such that
$\arun(0) = q$, $\arun(\aordinal) = q'$ and $\aordinal < \omega^n$. 
Given an ($\aordinal+1$)-path we write $\abstraction{\arun}$ to denote 
the triple $\triple{\arun(0)}{\ALL(\arun)}{\arun(\aordinal)}$.  
We define a family of relations containing the 
triples
of the form $\abstraction{\arun}$. 
Each relation $R_i$ below is therefore a subset of  
$Q \times \powerset{\abasis} \times Q$.

\begin{enumerate}[$\bullet$]
\item 
      $R_0 = \set{\triple{\alocation}{\alocation \cap 
             \alocation'}{\alocation'}: 
             \pair{\alocation}{\alocation'} \in \delta_{next}}$,
\item 
      For $i \in \Nat$,
      $$
      R'_i = \set{ \triple{\alocation_0}{\bigcap_{j=0}^{m} A_j}{\alocation_{m+1}} :
      \exists \ \alocation_0, \ldots, \alocation_{m+1},
                A_0, \ldots, A_{m} \ {\rm s.t.} \bigwedge_{j=0}^{m}  \triple{\alocation_j}{A_j}{\alocation_{j+1}} \in R_i
      }
      $$
\item 
     For $i \in \Nat$, $R_{i+1}$ is defined from $R_i'$ as follows: 
     $\triple{\alocation}{A}{\alocation'} \in R_{i+1}$ iff one of the conditions holds true:
     \begin{enumerate}[(1.):]
     \item $\triple{\alocation}{A}{\alocation'} \in R_{i}'$, 
     \item there exist $\triple{\alocation}{A'}{\alocation''} \in R_i'$ (2.1),
                $\triple{\alocation''}{\asetbis}{\alocation''} \in R_i'$ and
    a limit transition $\pair{\asetbis}{\alocation'} \in \delta_{lim}$ (2.2)
          such that $A = A' \cap \asetbis \cap \alocation'$. 
     \end{enumerate}
\end{enumerate}
The above numbering will be reused in Figure~\ref{figure-algorithm}. 

Let us first observe a few facts, whose proofs are by an easy verification.
\begin{enumerate}[(1)]
\item Whenever $\triple{\alocation}{A}{\alocation'} \in R_i$, $A \subseteq \alocation \cap \alocation'$.
\item Because  $R_i \subseteq R_{i+1}$ for all $i$, 
for some $N \leq 2^{3 \times \card{\abasis}} + 1$, 
$R_{N+1} = R_{N}$.
The bound $2^{3 \times \card{\abasis}} + 1$ takes simply into account that $Q \subseteq
\powerset{\abasis}$.
\end{enumerate}

In the sequel, for $n \geq 0$ and for $\triple{\alocation}{A}{\alocation'} \in \locations \times \powerset{\abasis} \times \locations$, 
we establish the  equivalence of the propositions below:
\begin{enumerate}[$\bullet$]
\item there is $\aordinal + 1 < \omega^{n+1}$ and an $(\aordinal+1)$-path $\arun$ such that
      $\abstraction{\arun} = \triple{\alocation}{A}{\alocation'}$,
\item $\triple{\alocation}{A}{\alocation'} \in R'_n$. 
\end{enumerate}

\iflpar
\begin{lemma} \label{lemma-nonemptiness-one}
(I) If $\triple{\alocation}{E,A}{\alocation'} \in R_{n}$, 
then there is an  ($\aordinal+1$)-path such that $\abstraction{\arun} = 
\triple{\alocation}{E,A}{\alocation'}$ and $\aordinal < \omega^n$. 
(II) Conversely, 
let $\arun: \aordinal + 1 \rightarrow Q$ be a path such that
$\aordinal < \omega^{n}$. Then $\abstraction{\arun} \in R_{n}'$.
\end{lemma}
\else
\begin{lem} \label{lemma-nonemptiness-one}
If $\triple{\alocation}{A}{\alocation'} \in R_{n}$, 
then there exist  $\aordinal < \omega^n$ and  an  ($\aordinal+1$)-path such that $\abstraction{\arun} = 
\triple{\alocation}{A}{\alocation'}$. 
\end{lem}

\iflong
\begin{proof}
The proof is by induction on $n$. For the base $n = 0$, the proof is by an easy
verification. In the induction step, suppose that 
$\triple{\alocation}{A}{\alocation'} \in R_{n+1}$.
First suppose that $\triple{\alocation}{A}{\alocation'} \in R_{n}'$, that is
there are $\alocation_0$, \ldots, $\alocation_{m+1}$,
                $A_0$, \ldots, $A_{m}$ such that $\bigwedge_{j=0}^{m}  \triple{\alocation_j}{A_j}{\alocation_{j+1}} \in R_{n}$,
$A = \bigcap_{j=0}^{m} A_j$, $\alocation_0 = \alocation$ and $\alocation_{m+1} = \alocation'$.
By induction hypothesis, for $i \in [0,m]$,  
there is a path $\arun_i: \aordinal_i + 1 \rightarrow Q$
such that 
$\abstraction{\arun_i} = \triple{\alocation_i}{A_i}{\alocation_{i+1}}$ and
$\aordinal_i < \omega^n$. Hence, $\arun_0 \cdot \cdots \cdot \arun_{m}$ is
a path of the desired form of length strictly less than $\omega^{n}$.


If $\triple{\alocation}{A}{\alocation'} \not \in R_{n}'$, then necessarily, 
by definition of $R_{n+1}$, 
there exist  $\pair{\asetbis}{\alocation'} \in \delta_{lim}$,
           $\triple{\alocation}{A'}{\alocation''} \in R_n'$ and $\triple{\alocation''}{\asetbis}{\alocation''} \in R_n'$ 
          such that $A = A' \cap \asetbis \cap \alocation'$. 

Hence, by definition of $R_n'$ and by induction hypothesis there is a path $\arun: \aordinal +1 \rightarrow
Q$ of length  strictly less than $\omega^n$ between $\alocation$ and $\alocation''$.
Similarly, there is a path $\arun': \aordinalbis +1 \rightarrow
Q$ of length  strictly less than $\omega^n$ between $\alocation''$ and $\alocation''$.
Observe that $\arun'' = \arun \cdot (\arun')^{\omega} \alocation'$ is a path of length strictly less than 
$\omega^{n+1}$, $\AL(\arun'') = \asetbis$ and $\abstraction{\arun''} = \triple{\alocation}{A}{\alocation'}$. 
\end{proof}
\else
The proof can be found in Appendix~\ref{section-proof-lemma-nonemptiness-one}. 
\fi 

Consequently, if  $\triple{\alocation}{A}{\alocation'} \in R'_n$, 
then  there is $\aordinal + 1 < \omega^{n+1}$ and an $(\aordinal+1)$-path $\arun$ such that
      $\abstraction{\arun} = \triple{\alocation}{A}{\alocation'}$. 
A converse result can also be established.

\begin{lem} \label{lemma-nonemptiness-two}
Let $\arun: \aordinal + 1 \rightarrow Q$ be a path such that
$\aordinal < \omega^{n}$. Then $\abstraction{\arun} \in R_{n}'$.  
\end{lem}

\iflong
\begin{proof}
The proof is by induction on $n$.  
The base case $n = 0$ is immediate. In the induction step, let $\arun$ be a path of length $\aordinal
< \omega^{n+1}$. If $\aordinal < \omega^{n}$, by induction hypothesis 
 $\triple{\arun(0)}{\ALL(\arun)}{\arun(\aordinal)} \in R_{n}'$ and therefore 
 $\triple{\arun(0)}{\ALL(\arun)}{\arun(\aordinal)} \in R_{n+1}$ 
since $R_n' \subseteq R_{n+1}$. 
Now suppose that $\aordinal = \omega^n \times m + \aordinalbis$
with $\aordinalbis < \omega^n$ and $m > 0$.  In order to show that  
$\triple{\arun(0)}{\ALL(\arun)}{\arun(\aordinal)} 
\in R_{n+1}'$ 
it is sufficient to consider the case $\aordinal = \omega^n$.
Indeed, $R'_{n+1}$ is closed under composition, i.e. if
$\triple{\alocation_0}{A_0}{\alocation'_0} \in R'_{n+1}$ and
$\triple{\alocation_0'}{A_1}{\alocation'_1} \in R'_{n+1}$,
then $\triple{\alocation_0}{A_0 \cap A_1}{\alocation'_1} \in R'_{n+1}$. 
So, suppose that $\arun$ is of length $\omega^n + 1$. By induction hypothesis, for every
$0 \leq i < i'$, 
 $\triple{\arun(\omega^{n-1} \times i)}{A_{i,i'}}{\arun(\omega^{n-1} \times i')}
\in R'_{n}$ for some $A_{i,i'}$. 
By Ramsey's Theorem, there are $0 < i_0 < i_1 < \ldots $ such that
$\triple{\arun(\omega^{n-1} \times i_k)}{A_{i_k,i_{k+1}}}{\arun(\omega^{n-1} \times i_{k+1})}$
is the same for all $k \geq 0$. 
Let $j = i_0$ and $j' = i_1$. 
By induction hypothesis, 
$\triple{\arun(\omega^{n-1} \times j)}{A_{j,j'}}{\arun(\omega^{n-1} \times j')} \in R'_n$ 
since the length of $\arun_{[\omega^{n-1} \times j, \omega^{n-1} \times j']}$ is strictly less than
$\omega^{n}$. Moreover, we have  $A_{j,j'} = 
\AL(\arun,\omega^n)$. So,  there exist $\triple{\arun(0)}{A'}{\alocation''} \in R_n'$ 
($A' = A_{0,j}$, $\alocation'' =  \arun(\omega^{n-1} \times j)$),
                $\triple{\alocation''}{\asetbis}{\alocation''} \in R_n'$
($\asetbis = A_{j,j'}$)
 and
    a limit transition $\pair{\asetbis}{\arun(\omega^n)} \in \delta_{lim}$ 
          such that $A = A' \cap \asetbis \cap \arun(\omega^n)$.
Consequently, $\triple{\arun(0)}{\ALL(\arun)}{\arun(\aordinal)} \in R_{n+1}'$.

\end{proof}
\else
The proof can be found in Appendix~\ref{section-proof-lemma-nonemptiness-two}. 
\fi


\fi
We provide below a first complexity result. 

\begin{lem} \label{lemma-nonemptiness-exptime}
The nonemptiness problem for simple ordinal automata with acceptance conditions
can be checked in 
 exponential time
in $\card{\abasis}$. 
\end{lem}

\proof
 Let $\aautomaton$ be of the form $\triple{\abasis}{Q, I, F,
\mathcal{F}}{\delta_{next},\delta_{lim}}$.
$\aautomaton$ has an accepting run iff either ($\mathbf{A}$) there are
$\alocation_0 \in I$, $\alocation_f \in F$ and 
 $A \subseteq \abasis$ such that $\triple{\alocation_0}{A}{\alocation_f} \in R_{n}'$  
for some $n$ or ($\mathbf{B}$) there are $\alocation_0 \in I$, and a
run $\arun$ from $\alocation_0$ such that $\AL(\arun)\in \mathcal{F}$. 
($\mathbf{A}$)
deals with accepting runs of length some successor ordinal,
whereas ($\mathbf{B}$) deals with accepting runs of length some limit
ordinal.

In order to check ($\mathbf{A}$), it
is sufficient to test for
 $\triple{\alocation_0}{A}{\alocation_f} \in I \times \powerset{\abasis} \times F$ 
 whether
 $\triple{\alocation_0}{A}{\alocation_f} \in R_{\card{\abasis} +
 3}' \subseteq R_{\card{\abasis} + 4}$. Since $\card{Q}$ is in
 $\mathcal{O}(2^{\card{\abasis}})$, computing
 $R_{\card{\abasis}+ 4}$ takes $\card{\abasis}+ 4$ steps that
requires polynomial time in $\length{\aautomaton}$ and exponential
time in $\card{\abasis}$, we obtain the desired result. Observe
that we can take advantage of the fact that computing the
transitive closure of a relation and the maximal strongly
connected components  can be done in polynomial time in the size
of the relations.

 By Ramsey's theorem,  ($\mathbf{B}$) is equivalent to the 
following condition:
  there are  $\alocation \in Q$, 
 $A \subseteq \abasis$, $A' \in \mathcal{F}$ and runs  $r_1$ and $r_2$
  such that $\abstraction{r_1} = \triple{\alocation_0}{A}{\alocation}$ and
 $\abstraction{r_2} = \triple{\alocation}{A'}{\alocation}$.

Hence. in order to check these, it   is enough to check whether
there are $\alocation_0 \in I$, $\alocation \in Q$ and 
 $A \subseteq \abasis$ such that
 $\triple{\alocation_0}{A}{\alocation} \in R_{\card{\abasis} +
 3}'$, $\triple{\alocation}{A'}{\alocation} \in
 R_{\card{\abasis} + 3}'$ and $A' \in \mathcal{F}$. 
This can be done  in  exponential time as for  ($\mathbf{A}$).\qed

As a corollary of Lemma~\ref{lemma-nonemptiness-exptime}, satisfiability for $\mainlogic$
is in $\exptime$. Moreover, this can be improved as shown in the proof
of Theorem~\ref{theorem-pspace} presented in Section~\ref{section-complexity}.

\iflpar
\else \subsection{A polynomial-space algorithm}
\label{section-polynomial-space}
\fi 

\iflong
We improve below the bound in Lemma~\ref{lemma-nonemptiness-exptime} by taking advantage
that the recursive depth is linear and  only paths of at most exponential length need
to be computed. 
\else
We improve below the bound in Lemma~\ref{lemma-nonemptiness-exptime}. 
\fi 

\begin{thm} \label{theorem-nonemptiness-pspace}
The nonemptiness problem for simple ordinal automata
can be checked in polynomial space in $\card{\abasis}$. 
\end{thm}

\iflong
\begin{proof}
\iflpar 
Following the proof of Lemma~\ref{lemma-nonemptiness-exptime},
\else By Lemma~\ref{lemma-nonemptiness-one} and Lemma~\ref{lemma-nonemptiness-two} and 
by the fact that for all $n \geq 0$, we have $R_n \subseteq R_{\card{\abasis} + 4}$,
we obtain that  
\fi 
$\aautomaton$ has an accepting run iff 
($\mathbf{A}$) there are $\alocation_0 \in I$, $\alocation_f \in F$ and 
$A \subseteq \abasis$ such that $\triple{\alocation_0}{A}{\alocation_f} \in R_{\card{\abasis}+ 3}$ 
or  
($\mathbf{B}$) there are $\alocation_0 \in I$, $\alocation \in Q$ and 
$A' \subseteq \abasis$ such that $\triple{\alocation_0}{A'}{\alocation} \in R_{\card{\abasis}+ 4}$, 
$\triple{\alocation}{A'}{\alocation} \in 
R_{\card{\abasis}+ 4}$ and $A' \in \mathcal{F}$. 
$\abasis$ denotes the basis of $\aautomaton$. 


The function PATH defined in Figure~\ref{figure-algorithm} checks recursively whether a triple belongs to $R_N$.
Typically, the specification is that there exists an accepting computation for
PATH($\aautomaton, \triple{\alocation}{A}{\alocation'}, N$) 
 iff $\triple{\alocation}{A}{\alocation'} \in R_N$ for the ordinal automaton $\aautomaton$.
It takes into account that the number of potential triples in $R_N$ is bounded. Observe that the algorithm is nondeterministic and any
guess that breaks some condition somewhere aborts the computation.

In order to check ($\mathbf{A}$), 
the  non-deterministic
algorithm  guesses 
$q_0 \in I$, $q_f \in F$ and 
$A \subseteq \abasis$ 
(encoded in polynomial space
in $\mathcal{O}(\card{\abasis})$ and test whether 
\begin{center}
PATH($\aautomaton,\triple{q_0}{A}{q_f}, \card{\abasis} + 4)$
\end{center}
returns {\tt true}. 
Condition ($\mathbf{B}$) admits a similar treatment. 
The non-deterministic algorithm PATH defined below works in polynomial space
in $\card{\abasis}$ assuming that the last argument is polynomial in $\card{\abasis}$ which is
the case with $\card{\abasis} + 4$. 

\begin{figure}
PATH($\aautomaton, \triple{\alocation}{A}{\alocation'}, N$)
\begin{itemize}
\itemsep 0 cm
\item If $N = 0$ then  (if  (either  
      $A \neq \alocation \cap \alocation'$ or $\pair{\alocation}{\alocation'} 
      \not \in \delta_{next}$) then {\tt abort} else return  {\tt true});
\item If $N > 0$ then go non-deterministically to  1. or 2.
      \begin{description}
      \itemsep 0 cm
      \item[(1.)] Guess on-the-fly a sequence 
                  $$\triple{\alocation_0}{A_0}{\alocation_1},  
           \triple{\alocation_1}{A_1}{\alocation_2}, \ldots,  
                 \triple{\alocation_m}{A_m}{\alocation_{m+1}}$$
                 such that
                 \begin{itemize}
                 \itemsep 0 cm
                 \item $m < 2^{3 \times \card{\abasis}+1} + 1$,
                 \item for $0 \leq i \leq m$, 
                       PATH($\aautomaton, 
                       \triple{\alocation_i}{A_i}{\alocation_{i+1}}, N-1$) returns {\tt true},
                 \item $A = \bigcap_j A_j$
                 \item $\alocation = \alocation_0$, $\alocation' = \alocation_{m+1}$;
                 \end{itemize}
        \item[(2.)] We guess here two long sequences:
                    \begin{description}
                    \item[(2.1)] Guess on-the-fly a sequence 
                  $$\triple{\alocation_0}{A_0}{\alocation_1},  
           \triple{\alocation_1}{A_1}{\alocation_2}, \ldots,  
                 \triple{\alocation_m}{A_m}{\alocation_{m+1}}$$
                 such that
                 \begin{itemize}
                 \itemsep 0 cm
                 \item $m < 2^{3 \times \card{\abasis}+1} + 1$,
                 \item for $0 \leq i \leq m$, 
                       PATH($\aautomaton, 
                       \triple{\alocation_i}{A_i}{\alocation_{i+1}}, N-1$) returns {\tt true},
                 \item $A' = \bigcap_j A_j$;
                 \item $\alocation_0 = \alocation$;
                 \end{itemize}
                    \item[(2.2)]
                    \iflpar 
                    Guess a limit transition $\pair{\asetbis}{\alocation'} \in \delta_{lim}$ and
                          on-the-fly a sequence 
                          $\triple{\alocation_0'}{E_0',A_0'}{\alocation_1'},  
           \triple{\alocation_1'}{E_1',A_1'}{\alocation_2'}, \ldots,  
                 \triple{\alocation_m'}{E_{m'}',A'_{m'}}{\alocation_{m'+1}'}$
                 such that
                     \else
                     Guess a limit transition $\pair{\asetbis}{\alocation'} \in \delta_{lim}$ and
                          on-the-fly a sequence 
                          $$\triple{\alocation_0'}{A_0'}{\alocation_1'},  
           \triple{\alocation_1'}{A_1'}{\alocation_2'}, \ldots,  
                 \triple{\alocation_m'}{A'_{m'}}{\alocation_{m'+1}'}$$
                 such that
                     \fi
                 \begin{itemize}
                 \itemsep 0 cm
                 \item $m' < 2^{3 \times \card{\abasis}+1}$,
                 \item for $0 \leq i \leq m'$, 
                       PATH($\aautomaton, 
                       \triple{\alocation_i'}{A_i'}{\alocation_{i+1}'}, N-1$) returns {\tt true},
                 \iflpar
                 \item $A = 
                             (A' \cap \alocation_{m'+1}') \cap \bigcap_j A_j'$,
                       $\asetbis = \bigcap_j A_j'$,
                       $\alocation_0' = \alocation_{m+1}$;
                  \else
                  \item $A = 
                             (A' \cap \alocation_{m'+1}') \cap \bigcap_j A_j'$,
                       $\asetbis = \bigcap_j A_j'$,
                 \item $\alocation_0' = \alocation_{m+1}$;
                  \fi
                 \end{itemize}
                    \end{description}
                  
      \end{description}
\item Return  {\tt true}.
\end{itemize}
\caption{Algorithm PATH}
\label{figure-algorithm}
\end{figure}

\noindent In (1.), guessing on-the-fly a long sequence means that only two consecutive triples are kept
in memory at any time. We introduce a counter that will guarantee that
$m < 2^{3 \times \card{\abasis}+1}$ and it requires only space in
$\mathcal{O}(\card{\abasis})$.  Moreover, in order to check
$A = \bigcap_j A_j$ we need two auxiliary
variables that bookkeep the 
$A_j$ computed so far. Similar techniques
are used in (2.) to guarantee that this non-deterministic algorithm requires
only polynomial space in  $\mathcal{O}(\card{\abasis} + N)$ (we only need 
more variables and steps). 
It is straightforward to show that
 PATH($\aautomaton, \triple{\alocation}{A}{\alocation'}, N$)
has a computation that returns {\tt true} (all the guesses were correct) iff
$\triple{\alocation}{A}{\alocation'} \in R_{N}$.
Finally, by using Savitch Theorem~\cite{Savitch70}, we can conclude that nonemptiness can be checked in 
deterministic polynomial space in $\card{\abasis}$. 

\end{proof}
\else
\iflpar
\begin{proof}

\end{proof}
\else The proof can be found in Appendix~\ref{section-proof-lemma-nonemptiness-pspace}. 
\fi 
\fi 

Observe that the algorithm in the proof of Theorem~\ref{theorem-nonemptiness-pspace}
runs in space $\mathcal{O}(\card{\abasis} \times (\card{\abasis} + log \ (\card{\locations}) +
log \ (\card{\delta_{lim}}) + log \ (\card{\delta_{next}}))$. 
Indeed,  the recursive depth is in $\mathcal{O}(\card{\abasis})$. 
This is certainly sufficient
to get forthcoming results about the complexity of $\mainlogic$. Nevertheless, 
 the exact complexity 
characterization of the nonemptiness problem is open. It seems unlikely that the problem can be solved
in \nlogspace.

\newcommand{\tuple}[1]{\left(#1\right)}
\def\om{\omega}
\newtheorem{observation}{Observation}

\section{Complexity of Satisfiability Problems}
\label{section-complexity}

\noindent We establish new complexity results for problems related
to $\mainlogic$ satisfiability which follow from the intermediate results we have established
so far.

%
%

\subsection{Complexity of \texorpdfstring{$\mainlogic$}{LTL(U,S)}}


Here is the main result of the paper.

\begin{thm} \label{theorem-pspace}
The satisfiability problem for $\mainlogic$ over the class of
ordinals is \pspace-complete.
\end{thm}

\begin{proof} By Theorem~\ref{theorem-preliminariesbis}(II), 
a formula is satisfiable iff it is satisfiable on some model of countable length. 
By Lemma~\ref{lemma-correctness}, given
a formula $\aformula$ in $\mainlogic$, there is an automaton
$\aautomaton_{\aformula}$ whose accepting runs correspond exactly to models of $\aformula$.
In order to check nonemptiness of $\aautomaton_{\aformula}$, we do not build it explicitly (as usual)
but we run the algorithm from the proof of Theorem~\ref{theorem-nonemptiness-pspace}
and we compute the locations, and transition relations of $\aautomaton_{\aformula}$
on demand. Hence, we obtain a polynomial space non-deterministic algorithm since
the cardinality of the basis of $\aautomaton_{\aformula}$ is in $\mathcal{O}(\length{\aformula})$
and checking whether a subset of $\abasis$ is a location of
$\aautomaton_{\aformula}$ or
$\pair{\alocation}{\alocation'} \in \delta_{next}$
or $\pair{\asetbis}{\alocation} \in \delta_{lim}$ can be done in
polynomial space in  $\mathcal{O}(\length{\aformula})$.
Again by Savitch Theorem~\cite{Savitch70}, we get that the satisfiability problem  for $\mainlogic$
is in \pspace. The \pspace \ lower bound can be easily shown inherited from LTL.
\end{proof}

\iflpar
\else
Our procedure to show the \pspace \ upper bound is not optimal and it is subject to
many refinements but it is sufficient for our needs. For instance, it is possible
to have as a base set for $\aautomaton_{\aformula}$ the subset of $\subf{\aformula}$ made
 of until or since subformulae and propositional variables.
Indeed, the implicit presence of other subformulae can be deduced
thanks to maximal consistency. This refinement  possibly decreases
the length of the small models.
\fi

Due to Kamp's Theorem~\cite{Kamp68}, we get the following corollary.

\begin{cor} \label{corollary-pspace}
 Let $\ltl(\until, \since, \anoperator_1, \ldots,
\anoperator_k)$ be an extension of $\mainlogic$ with $k$
first-order definable temporal operators.

Then the satisfiability
problem for the logic $\ltl(\until, \since, \anoperator_1, \ldots,
\anoperator_k)$ over the class of ordinals is in \pspace.
\end{cor}

Indeed, every formula $\anoperator_i(\avarprop_1,
\ldots, \avarprop_{n_i})$ encoded as a DAG can be translated into an equivalent
 formula in $\mainlogic$ encoded as a DAG over the propositional
variables $\avarprop_1, \ldots, \avarprop_{n_i}$. Since
$\anoperator_1, \ldots, \anoperator_k$ and their definition in
$\mainlogic$ are constants of $\ltl(\until, \since, \anoperator_1,
\ldots, \anoperator_k)$, we obtain a translation in polynomial-time
(with our definition for the size of formulae).

\subsection{A family of satisfiability problems}

The satisfiability problem for $\mainlogic$ asks for the existence
of a model for a given formula. A natural variant of this problem
consists in fixing the length of the models in advance as for LTL.  The
satisfiability problem for $\mainlogic$ over $\aordinal$-models,
noted $\SAT(\aordinal,\mainlogic)$, is defined as follows:
\iflpar
given a formula  $\aformula$ in $\mainlogic$, is 
$\aformula$ satisfiable over an $\aordinal$-model?
\else 
\begin{description}
\itemsep 0 cm
\item[input:] a formula $\aformula$ in $\mainlogic$;
\item[question:] Is  $\aformula$ satisfiable over an $\aordinal$-model?
\end{description}
\fi 
In this subsection we prove 
that $\SAT(\aordinal,\mainlogic)$ is in \pspace \ for every
countable ordinal $\aordinal$.
\newcommand{\Def}{\mathit{Def}}
First we consider  the case of ordinals strictly less than $\omega^\omega$. Let us establish 
that for every $\aordinal <\omega^\omega$ there is a
formula $\defordinal{\aordinal}$ in $\mainlogic$  with the truth constant
$\top$ (no propositional variable) such that
for every $\aordinalbis$-model $\amodel$,
we have $\amodel, 0 \models  \defordinal{\aordinal}$ iff
$\aordinalbis = \aordinal$.


\begin{lem} \label{lemma-length}
Given an ordinal $0 < \aordinal = \omega^{k_1} a_{k_1} + \cdots 
\omega^{k_m} a_{k_m} < \omega^{\omega}$ with $ k_1 > \ldots > k_m \geq 0$ and 
$a_{k_1}, \ldots a_{k_m} \in \Nat \setminus \set{0}$, there is a formula 
$\defordinal{\aordinal}$ in $\mainlogic$
of linear size in $\sum_i (k_i \times a_{k_i})$  such that for any model
$\amodel$, we have
$\amodel, 0 \models \defordinal{\aordinal}$ iff $\amodel$ is of length $\aordinal$. 
\end{lem}

\iflpar
\else
\begin{proof} We define a family $(\aformulater_i)_{i \geq 0}$ such that
for all $\aordinal$-models $\amodel$  and $\aordinalbis < \aordinal$,
we have $\amodel, \aordinalbis \models \aformulater_i$ iff $\aordinalbis$ is a multiple of
$\omega^i$. We set $\aformulater_0 = \top$ and by induction
$\aformulater_{i+1} = \aformulater_i \wedge \neg (\neg \aformulater_i \since \aformulater_i)$.
Observe that $\length{\aformulater_i}$ is polynomial in $i$ since it is defined as 
the cardinality
of $\subf{\aformulater_i}$. 
Now let us define $\defordinal{\aordinal}$  as $t(\aordinal)$ defined recursively below:
\begin{enumerate}[$\bullet$]
\item $t(1) = \neg \strictsometimes \top$, $t(n) = \next t(n-1)$ for $n >1$,
\item $t(\omega^{k_1} a_{k_1} + \cdots +
\omega^{k_m} a_{k_m}) = \neg \aformulater_{k_1} \until (\aformulater_{k_1} \wedge 
      t(\omega^{k_1} (a_{k_1} -1) + \cdots + 
\omega^{k_m} a_{k_m}))$ with $k_1 > 0$ and ($a_{k_1} \geq 2$ or $m > 1$),
\item $t(\omega) = \strictalways \previous \top \wedge \strictsometimes \top \wedge \strictalways \next \top$,
\item $t(\omega^{k_1}) = \strictalways \neg \aformulater_{k_1} \wedge
                         \always \strictsometimes \aformulater_{k_1-1}$ with $k_1 > 1$.
\end{enumerate}
The size of $\defordinal{\aordinal}$ 
is in $\mathcal{O}(\sum_i (k_i \times a_{k_i}))$. 
\end{proof}

We are now in position to state the following result. 

\begin{cor} \label{theorem-sat-family}
For every $\aordinal<\omega^\omega$, the problem
$\SAT(\aordinal,\mainlogic)$
is in  \pspace.
\end{cor}

\begin{proof}
 $\aformula$ has a $\aordinal$-model
 iff
 $\aformulabis =
      \aformula \wedge \defordinal{\aordinal}$
is satisfiable over the class of ordinals.
Thanks to Lemma~\ref{lemma-length} and Theorem~\ref{theorem-pspace},
we obtain the \pspace \ upper bound. 
\end{proof}

Now we consider the case of a countable ordinal $\alpha\geq
\omega^\omega$. Let $\aordinal'$ be the unique ordinal strictly
less than $\omega^{\omega}$ such that $\aordinal = \omega^{\omega}
\times \aordinalter + \aordinal'$ for some ordinal $\aordinalter$.
Note that for every $k$,  $\truncation{k} (\aordinal)=
\truncation{k} (\omega^k+\aordinal')<\omega^\omega$.
By  Lemma~\ref{lemma-pumping-formulae}(II), $\aformula$ has an $\aordinal$-model
iff $\aformula$ has a $\aordinal_{\length{\aformula}}$-model
with
$\aordinal_{\length{\aformula}} =
\truncation{\length{\aformula}+
\plusconstant}(\aordinal)=\truncation{\length{\aformula}+
\plusconstant}(\omega^{\length{\aformula}+
\plusconstant}+\aordinal')$. Hence, $\aformula$ has an $\aordinal$-model
iff $\aformula\wedge \defordinal{\aordinal_{\length{\aformula}}}$
is satisfiable (over
the class of countable ordinals). Since the size of
$\defordinal{\aordinal_{\length{\aformula}}}$
is polynomial in the
size of $\aformula$, we derive from Theorem \ref{theorem-pspace}
the following result.

\begin{cor} \label{theorem-sat-family-big}
For every countable $\aordinal \geq \omega^\omega$, the problem
$\SAT(\aordinal,\mainlogic)$
 is in  \pspace.
\end{cor}

 Corollaries~\ref{theorem-sat-family},~\ref{theorem-sat-family-big} and the arguments similar 
to the arguments in
 the proof of Corollary \ref{corollary-pspace} imply the result below.

\begin{thm} \label{theorem-finite-operators}
The satisfiability problem
for $\ltl(\anoperator_1, \ldots, \anoperator_k)$ restricted to
$\aordinal$-models is in \pspace, for every  finite set $\set{\anoperator_1, \ldots,
\anoperator_k}$ of first-order definable temporal operators and
for every countable ordinal $\aordinal$.
\end{thm}

Observe that $\aordinal$ finite implies
$\SAT(\aordinal,\ltl(\anoperator_1, \ldots, \anoperator_k))$ is
\np-complete, otherwise
\pspace-hardness for $\SAT(\aordinal,\mainlogic)$  follows from
\pspace-completeness of  $\SAT(\om,\mainlogic)$.

\subsection{Uniform satisfiability}



B\"{u}chi (see, e.g., \cite{Buchi&Siefkes73}) has shown that there
is a \emph{finite} amount of data concerning any countable ordinal
that determines its monadic theory.
\begin{defi}[Code of an ordinal]\label{dfn:code of an ordinal} Let $\alpha$ be a countable ordinal and
let $m $ be in $[1,\omega]$.
\begin{enumerate}[(1)]
\item Write $\alpha=\om^m\alpha' + \zeta$ with $\zeta<\om^m$ (this can be done in a unique way),
and let
$$p_m(\alpha):= \left\{ \begin{array}{lll}
-2 & \textrm{if $\alpha'=0$ }\\
-1 & \textrm{if $0<\alpha'<\om_1$}\\
\end{array} \right..$$
\item If $\zeta\ne 0$, write $\zeta = \sum_{i\le n}\om^{n-i}\cdot a_{n-i}$
where $a_i\in \om$ for $i\le n$ and $a_n\ne 0$ (this can be done in
a unique way), and let $t_m(\alpha):=\tuple{a_n,\ldots,a_0}$. If
$\zeta=0$, let $t_m(\alpha)=-3$.
\item The $m$-\emph{code} of $\alpha$ is the pair $\tuple{ p_m(\alpha),t_m(\alpha)}$.
\end{enumerate}
\end{defi}

The following is implicit in \cite{Buchi&Siefkes73}.

\begin{thm}[Code Theorem]\label{thm:uniform decidability of MTh below om_1}
There is an algorithm that, given a monadic second-order sentence
$\aformula$ and the $\om$-\emph{code} of a countable ordinal  $\aordinal$,
determines whether $\pair{\aordinal}{<}\models \aformula$.
\end{thm}

Lemma~\ref{lemma-pumping-formulae} can be rephrased as ``the
($\length{\aformula}+2$)-code of an ordinal  $\alpha$ determines
whether  $\aformula$ has a model of length $\alpha$".

Let $C= \tuple{b,a_n, \dots a_0}$  be an $m$-code. Its size is
defined as $n+a_0+a_1+\dots +a_n$. It is clear that for $m_1<m_2$
the $m_2$-code of an ordinal determines its  $m_1$-code and there is
a linear-time algorithm, that given $m_2$-code of an ordinal and
$m_1<m_2$ computes the $m_1$-code of the ordinal.

The arguments used in the proof of Corollary~\ref{theorem-sat-family-big} show the following theorem.

\begin{thm}[Uniform Satisfiability]\label{thm:uniform} \ 
\begin{enumerate}[\em(I):]
\item There is a polynomial-space
     algorithm that, given an $\mainlogic$ formula
 $\aformula$ and the $\om$-\emph{code} of a countable
 ordinal  $\aordinal$,
determines whether  $\aformula$ has an $\alpha$-model.
\item There is a polynomial-space algorithm that, given an $\mainlogic$ formula
 $\aformula$ and the    ($\length{\aformula}+2$)-code of a
countable ordinal  $\alpha$,
determines whether $\aformula$ has an $\alpha$-model.
\end{enumerate}
\end{thm}

\section{Related Work}
\label{section-related-work}

\noindent In this section, we compare our results with those from the literature related to satisfiability.
It is worth noting that an axiomatization of  
$\mainlogic$ over ordinals can be found in~\cite{Venema93}.
Nevertheless, the concern in this above-mentioned paper is quite different from ours. 

\subsection{Comparison with Rohde's thesis}

In~\cite{Rohde97}, it is shown that an uniform satisfiability problem
for temporal logic
with until (and without since) can be solved in exponential time (flows of time are countable ordinals).
The inputs of this
problem are a formula in $\ltl(\until)$ and the representation of a countable ordinal.
The satisfiability problem is also shown in \exptime.
In order to obtain this upper bound, formulae are shown equivalent to alternating
automata and a reduction from alternating automata into a specific subclass of non-deterministic
automata is given. Finally, a procedure for testing nonemptiness is provided.
Here are the similarities between~\cite{Rohde97} and our results.
\begin{enumerate}[(1)]
\item We also follow an automata-based approach and
      the class of non-deterministic automata in~\cite{Rohde97} and ours have
a structured set of locations and limit transitions use elements that are true from some position.
\item Existence of $\aordinal$-paths in the automata depends on some truncation of
$\aordinal$.
\item The logical decision problems can be solved in exponential time.
\end{enumerate}
However, our work improves  some results from~\cite{Rohde97}.
\begin{enumerate}[(1)]
\item Our temporal logic includes the until and since operators (instead of until
      only) and it is therefore as expressive as first-order logic.
\item We establish a tight \pspace \ upper bound (instead of \exptime)
      thanks to the introduction of  simple ordinal automata.
\item Our proofs are  shorter and more transparent (instead of the lengthy developments
found in~\cite{Rohde97}).
\end{enumerate}

Consequently, the developments from~\cite{Rohde97} and ours follow the same approach
with different definitions for automata, different intermediate lemmas and distinct
final complexity bounds. On the other hand, the structure of the whole proof to obtain
the main complexity bounds is similar.

\subsection{LTL over other classes of linear orderings}
\def\leftl{\mathit{leftlim}}
\def\rightl{\mathit{rightlim}}

Even though the  results for linear-time temporal logics
from~\cite{Reynolds03,Reynolds??} involve distinct models, our
automata-based approach has similarities
with these works that uses a different proof
method, namely mosaics. Indeed, equivalence classes of the
relation $\sim$ between runs of length a successor ordinal roughly
correspond to mosaics from~\cite{Reynolds03}.
We recall the main results below.

\begin{thm}\label{th:reynolds} \ 
\iflong
 \begin{enumerate}[\em(I):]
\itemsep 0 cm
\item \cite{Reynolds??} The satisfiability
problem for the temporal logic with until and since
over the reals is \pspace-complete.
\item  \cite{Reynolds03} The satisfiability
problem for $\ltl(\until)$  over the class of all linear orders is
\pspace-complete.
\end{enumerate}
\else
(I) The satisfiability
problem for the temporal logic with until and since
over the reals is \pspace-complete.
(II) The satisfiability
problem for $\ltl(\until)$  over the class of all linear orders is
\pspace-complete.
\fi
\end{thm}


\noindent The proofs in~\cite{Reynolds03,Reynolds??}  are much more involved
than our proofs since the orders  are more
complex than the class of ordinals.
 Moreover, 
a recent work~\cite{Cristau09} has established that $\mainlogic$ over the class of linear orderings
has an elementary complexity by using transducers as done in~\cite{Michel84} for standard LTL. 
More precisely, satisfiability for 
$\mainlogic$ augmented with future and past Stavi operators is in 
\twoexpspace~\cite{Cristau09}. 
Nevertheless, complexity of $\mainlogic$ over the class of linear orderings has been recently 
solved: for any temporal logic with a finite set of modalities definable
in the existential fragment of second-order logic has a \pspace \ satisfiability
problem over the class of linear orderings~\cite{Rabinovich10,Rabinovich10b} (see also~\cite{Reynolds10}).
Moreover, observe  that $\mainlogic$ over the reals has been recently shown in \pspace~in~\cite{Reynolds10},
which allows us to obtain in a different way that $\mainlogic$ over the countable ordinals is in \pspace~(see the full arguments
in~\cite[Section 13]{Rabinovich10}).

\subsection{Quantitative temporal operators}

In this section, we show that the main results
from~\cite{Demri&Nowak07} are subsumed by the current paper.
We also solve an open problem from~\cite{Cachat06,Demri&Nowak07}.
For every fixed countable ordinal $\aordinal \leq \omega$,
let us introduce the logic LTL($\mathcal{O}_{\aordinal}$)
where the set of temporal operators $\mathcal{O}_{\aordinal}$ is defined as follows:
$
\set{\next^{\aordinalbis}: \aordinalbis < \omega^{\aordinal}}
\cup
\set{\until^{\aordinalbis}: \aordinalbis \leq \omega^{\aordinal}}.
$
The models of $\ltl(\mathcal{O}_{\aordinal})$ as those of $\mainlogic$
and the formulae of $\ltl(\mathcal{O}_{\aordinal})$ are precisely defined by:
\iflong
$$
\aformula ::= \avarprop \ \mid \
              \neg \aformula \ \mid \
              \aformula_1 \wedge \aformula_2 \ \mid \
              \next^{\aordinalbis} \aformula \ \mid \
              \aformula_1 \until^{\aordinalbis} \aformula_2.
$$
\else
$
\aformula ::= \avarprop \ \mid \
              \neg \aformula \ \mid \
              \aformula_1 \wedge \aformula_2 \ \mid \
              \next^{\aordinalbis} \aformula \ \mid \
              \aformula_1 \until^{\aordinalbis'} \aformula_2.
$
\fi
The satisfaction relation is inductively defined below where $\amodel$ is a model
for LTL($\mathcal{O}_{\aordinal}$) (we omit the obvious clauses):
\begin{enumerate}[$\bullet$]
\item $\amodel, \aordinalbis \models \next^{\aordinalbis'} \aformula$ iff
       $\aordinalbis +  \aordinalbis'$ is a position of $\amodel$ and
       $\amodel, \aordinalbis +  \aordinalbis' \models \aformula$,
\item $\amodel, \aordinalbis \models \aformula_1 \until^{\aordinalbis'} \aformula_2$ iff
      there is $\aordinalter \in (0,\aordinalbis')$ such that
      $\aordinalbis +  \aordinalter$ is a position of $\amodel$,
      we have $\amodel, \aordinalbis +  \aordinalter \models \aformula_2$ and
      for every $\aordinalter' \in (0,\aordinalter)$,
      we have $\amodel, \aordinalbis +  \aordinalter' \models \aformula_1$.
\end{enumerate}

\medskip\noindent The satisfiability problem for  LTL($\mathcal{O}_{\aordinal}$) consists in determining,
given a formula $\aformula$, whether there is a model $\amodel$ such that
$\amodel, 0 \models \aformula$.
The main results of~\cite{Cachat06,Demri&Nowak07} are the following:
\iflong
\begin{enumerate}[(1)]
\item For every $k \in \Nat \setminus \set{0}$, the satisfiability problem for
       LTL($\mathcal{O}_{k}$) restricted to models of length $\omega^k$
       is \pspace-complete when the natural numbers occurring in formulae are encoded
       in unary. With binary representation, it becomes \expspace-hard (mainly because
       a temporal operator $\next^{2^n}$ is helpful to specify concisely the cell contents of
       exponential-space Turing machines). 
\item  LTL($\mathcal{O}_{\omega}$) restricted to models of length $\omega^{\omega}$
       is decidable.
\end{enumerate}
\else
for every $k \geq 1$, the satisfiability problem for
       LTL($\mathcal{O}_{k}$) restricted to models of length $\omega^k$
       is
\pspace-complete and
LTL($\mathcal{O}_{\omega}$) restricted to models of length $\omega^{\omega}$
       is decidable.
\fi

\medskip\noindent Observe that LTL($\mathcal{O}_{k}$) cannot express
the temporal operator $\until$ over the class of countable ordinals
but it can do it on models of length $\omega^k$.  Hence, each logic
$\ltl(\mathcal{O}_{k})$ is less expressive than $\mainlogic$.

 Moreover, it is easy to  show that
for every $\aordinal \leq \omega$, the logic
$\ltl(\mathcal{O}_{\aordinal})$ is expressively equivalent (over the
class of countable ordinals)  to its sublogic over the following set
$\mathcal{O}_{\aordinal}'$ of temporal operators:
$$\mathcal{O}_{\aordinal}'=
\set{\next^{\omega^i}: \omega^i < \omega^{\aordinal}, i \in \Nat}
\cup \set{\until^{\omega^{\aordinalbis}}: \omega^\aordinalbis \leq
\omega^{\aordinal}, \aordinalbis \leq \omega}.$$ This set is finite
when $\aordinal$ is finite. Moreover, there is a linear-time (and
logarithmic space) meaning preserving translation from
$\ltl(\mathcal{O}_{\aordinal})$ into
$\ltl(\mathcal{O}_{\aordinal}')$.

\iflpar
\else 
Let us translate $\aformula$ in $\ltl(\mathcal{O}_{\omega}')$
into a formula $t(\aformula)$ in $\mainlogic$
homomorphically for the Boolean operators and such that the propositional
variables remain unchanged. Here are the remaining clauses of translation:
\begin{enumerate}[$\bullet$]
\item $t(\aformulabis_1 \until^{\omega^i} \aformulabis_2)
       = (\neg \aformulater_i
      \wedge t(\aformulabis_1)) \until (\neg \aformulater_i
       \wedge t(\aformulabis_2))$, $t(\aformulabis_1 \until^{\omega^{\omega}} \aformulabis_2)
       =  t(\aformulabis_1) \until t(\aformulabis_2)$,
\item $t(\next^{\omega^i} \aformulabis_1) =
       \neg \aformulater_i  \until
      (\aformulater_i \wedge t(\aformulabis_1))$.
\end{enumerate}
The formula $\aformulater_i$ is defined in the proof of Lemma~\ref{lemma-length}.
The following result is easy to show.

\begin{lem}\hfill
\begin{enumerate}[\em(I):]
\item Let $\aformula$ be  in $\ltl(\mathcal{O}_{\omega}' \setminus
\set{\until^{\omega^{\omega}}})$.
$t(\aformula)$ is equivalent to $\aformula$ over the class of countable ordinals, i.e.
for all  $\aordinal$-models $\amodel$  and $\aordinalbis < \aordinal$, we have
      $\amodel, \aordinalbis \models \aformula$ iff
      $\amodel, \aordinalbis \models t(\aformula)$.
\item   Let $\aformula$ be  in $\ltl(\mathcal{O}_{\omega}')$.
For all  $\omega^{\omega}$-models $\amodel$ and
 $\aordinalbis < \omega^{\omega}$,
we have      $\amodel, \aordinalbis \models \aformula$ iff
      $\amodel, \aordinalbis \models t(\aformula)$.
Moreover, $\length{t(\aformula)}$ is linear in $\length{\aformula}$.
\end{enumerate}
\end{lem}
\fi 

\noindent (I) is essentially based on the properties of formulae $\aformulater_i$ and on the exclusion
of $\until^{\omega^{\omega}}$. (II) simply takes advantage of the fact that for the 
$\omega^{\omega}$-models, $\until$ and $\until^{\omega^{\omega}}$ are obviously equivalent. 

We obtain  alternative proofs for known results and we get new results.

\begin{thm} For every $k \in \Nat \setminus \set{0}$,
\begin{enumerate}[\em(I):]
\item the satisfiability problem for $\ltl(\mathcal{O}_{k})$
over $\omega^k$-models
is in \pspace \ with unary encoding of natural numbers,
\item the satisfiability problem for $\ltl(\mathcal{O}_{k}')$ restricted
           to $\omega^k$-models  is \pspace-complete,
\item for every countable infinite ordinal $\aordinal$,
             the satisfiability problem for  $\ltl(\mathcal{O}_{k}')$
             restricted to $\aordinal$-models is \pspace-complete.
\end{enumerate}
\end{thm}

\noindent (III) is an instance of Theorem~\ref{theorem-finite-operators}.
(II) is an instance of (III) (with unary encoding of natural numbers). 
(I) can be shown by observing that
there is a logarithmic space meaning preserving translation from
$\ltl(\mathcal{O}_{k})$
          to  $\ltl(\mathcal{O}_{k}')$. (I) is the  main result of~\cite{Demri&Nowak07}
with the unary encoding of natural numbers occurring in ordinal expressions.
Finally, the corollary below improves the non-elementary bounds
obtained in~\cite{Cachat06,Demri&Nowak07}  for
$\ltl(\mathcal{O}_{\omega})$  by reducing this temporal logic to
the monadic second order logics, and then to Buchi ordinal
automata.

\begin{cor} Satisfiability for  $\ltl(\mathcal{O}_{\omega})$ over the class
of $\omega^{\omega}$-models is \pspace-com\-ple\-te with unary encoding of natural numbers 
in formulae.
\end{cor}



\section{Conclusion}
\label{section-conclusion}

\noindent In the paper, we have shown that the linear-time temporal
logic with until and since over the class of ordinals, namely
$\mainlogic$ has a \pspace-complete satisfiability problem.  Due to
Kamp's Theorem~\cite{Kamp68}, we know that $\mainlogic$ is a
fundamental temporal logic since it is as expressive as first-order
logic over the class of ordinals.  In order to establish this tight
complexity characterization, we have introduced the class of simple
ordinal automata.  This class of automata is more structured than
usual ordinal automata
and the sets of locations have some structural properties, typically it is a subset
of the powerset of some set (herein called the basis).
As a consequence, we are also able to improve some results from~\cite{Rohde97,Demri&Nowak07}.
For instance the uniform satisfiability problem is \pspace-complete and we obtain
alternative proofs for results in~\cite{Demri&Nowak07}. 
Recent results about the polynomial space upper bound for LTL over various
classes of linear orderings can be found in~\cite{Rabinovich10,Rabinovich10b} by using the so-called
composition technique and the automata-based technique used in
this paper.


{\bf Acknowledgments:} We would like to thank the anonymous referees for 
helpful suggestions and remarks.




\iflpar
\else
\iflong
\else
\newpage
\appendix
\input{appendix}
\fi
\fi

\end{document}